\documentclass[11pt,a4paper]{article}
\pdfoutput=1
\usepackage[margin=1in]{geometry}

\usepackage[utf8]{inputenc}
\usepackage[english]{babel}
\usepackage[numbers,sort&compress]{natbib}
\usepackage{microtype}
\usepackage{xcolor}
\usepackage{graphicx}
\usepackage{amsmath}
\usepackage{amsthm}
\usepackage{amssymb}
\usepackage{booktabs}
\usepackage[unicode]{hyperref}
\usepackage{doi}
\usepackage{framed}
\usepackage{cleveref}
\usepackage{xspace}

\usepackage{algorithm2e}
\usepackage[noend]{algpseudocode}
\usepackage{thm-restate}

\usepackage[inline]{enumitem}
\setlist[itemize]{label=--}
\setlist[enumerate]{label=(\arabic*),labelindent=\parindent,leftmargin=*}

\newcommand{\logstar}{\log^{*}}
\newcommand{\eps}{\varepsilon}
\newcommand{\lovasz}{Lov\'{a}sz\xspace}

\newcommand{\proposalAnswerParameter}{\textsf{ProposalParameter}}
\newcommand{\payPerKill}{\textsf{PayPerKill}}
\newcommand{\phases}{\textsf{Phases}}
\newcommand{\totalTokens}{\textsf{totalTokens}}
\newcommand{\levels}{\textsf{Levels}}
\newcommand{\steps}{\textsf{Steps}}

\hypersetup{
    colorlinks=true,
    linkcolor=black,
    citecolor=black,
    filecolor=black,
    urlcolor=[rgb]{0,0.1,0.5},
    pdftitle={Efficient CONGEST Algorithms for the \lovasz Local Lemma},
    pdfauthor={Yannic Maus and Jara Uitto}
}

\newtheorem{theorem}{Theorem}[section]
\newtheorem{lemma}[theorem]{Lemma}
\newtheorem{corollary}[theorem]{Corollary}
\newtheorem{definition}[theorem]{Definition}
\newtheorem{claim}[theorem]{Claim}
\newtheorem{observation}[theorem]{Observation}

\newcommand{\LOCAL}{\ensuremath{\mathsf{LOCAL}}\xspace}
\newcommand{\CONGEST}{\ensuremath{\mathsf{CONGEST}}\xspace}

\newcommand{\myemail}[1]{\,$\cdot$\, {\small #1}}
\newcommand{\myaff}[1]{\,$\cdot$\, {\small #1}\par\medskip}

\DeclareMathOperator{\poly}{poly}

\newenvironment{myabstract}
{\list{}{\listparindent 1.5em%
        \itemindent    \listparindent
        \leftmargin    1cm
        \rightmargin   1cm
        \parsep        0pt}%
    \item\relax}
{\endlist}

\newenvironment{mycover}
{\list{}{\listparindent 0pt
        \itemindent    \listparindent
        \leftmargin    1cm
        \rightmargin   1cm
        \parsep        0pt}%
    \raggedright
    \item\relax}
{\endlist}

\begin{document}

\begin{mycover}
  {\huge\bfseries\boldmath Efficient CONGEST Algorithms for the \lovasz Local Lemma \par}
\bigskip
\bigskip
\bigskip
\textbf{Yannic Maus}
\myemail{yannic.maus@cs.technion.ac.il}
\myaff{Technion}

\textbf{Jara Uitto}
\myemail{jara.uitto@aalto.fi}
\myaff{Aalto University}

\bigskip

\end{mycover}
\begin{myabstract}
\noindent\textbf{Abstract.}
	We present a $\poly \log \log n$ time randomized \CONGEST algorithm for a natural class of \lovasz Local Lemma (LLL) instances on constant degree graphs.
	This implies, among other things, that there are no LCL problems with randomized complexity between $\log n$ and $\poly \log \log n$.
	Furthermore, we provide extensions to the network decomposition algorithms given in the recent breakthrough by Rozho\v{n} and Ghaffari [STOC2020] and the follow up by Ghaffari, Grunau, and Rozho\v{n} [SODA2021].
	In particular, we show how to obtain a large distance separated weak network decomposition with a negligible dependency on the range of unique identifiers.
\end{myabstract}

\thispagestyle{empty}
\setcounter{page}{0}
\newpage


\section{Introduction}

 Our main contribution is a $\poly\log\log n$ round randomized distributed \CONGEST algorithm to solve a natural class of \lovasz Local Lemma (LLL) instances  on constant degree graphs. Among several other applications, e.g., various defective graph coloring variants,  this implies that there is no LCL  problem (locally checkable labeling problem) with randomized complexity strictly between $\poly\log\log n$ and $\Omega(\log n)$, which together with the results of~\cite{Balliu2021} implies that the \emph{known world} of complexity classes in the sublogarithmic regime in the \CONGEST and \LOCAL model are almost identical.  As a side effect of our techniques we extend the understanding of the computation of network decompositions in the \CONGEST model. We now explain our contributions on LLL and LCLs in more detail; in the second half of the introduction we explain our techniques, contributions on network decomposition and how they relate to results in the \LOCAL model.
 
 We work in the \CONGEST model of distributed computing, where a network of computational devices is abstracted as an $n$-node graph, where each node corresponds to a computational unit.
 In each synchronous round, the nodes can send messages of size $b = O(\log n)$ to their neighbors.
 In the end of the computation, each node is responsible of outputting its own part of the output/solution, e.g., its color in a graph coloring problem.
 The \LOCAL model is otherwise the same, except that there is no bound on the message size $b$.
 
\paragraph*{Background on LLL and LCLs in \LOCAL.} An \emph{instance} of the \emph{\lovasz Local Lemma problem} is formed by a set of variables and a set of bad event $\mathcal{E}_1,\ldots,\mathcal{E}_n$ that depend on the variables. The famous \lovasz Local Lemma \cite{LLL73} states that if the probability of each event is upper bounded by $p$, each event only shares variables with at most $d$ other events and $epd<1$ holds, then there exists  an assignment to the variables that avoids all bad events. An example of a problem that can be solved via LLL is the \emph{sinkless orientation} problem on graphs with minimum degree $4$. In the sinkless orientation problem the objective is to orient the edges of a graph such that every vertex has at least one outgoing edge. This can be modeled by an LLL as follows: Orient each edge randomly ---each edge represents a variable--- and introduce a bad event for each vertex that holds if and only if all edges are oriented towards it. One obtains $p=2^{-d}$ and $d\geq 4$ such that $epd<1$ holds.  
  The \lovasz Local Lemma has had a huge success in theory of computation. One highlight is the beautiful and simple parallel algorithm by Moser and Tardos~\cite{MoserTardos10}.   
In the distributed version of the problem each event corresponds to a node in the communication network and the assignment of the random variables, that are potentially shared by many events, is done by the corresponding nodes.
The algorithm by Moser and Tardos immediately yields a randomized $O(\log^2 n)$  LLL algorithm in the \LOCAL model. Often one has stronger guarantees on the relation of $p$ and $d$, e.g., a  polynomial criterion $epd^2<1$ (or even larger exponents) instead of only $epd<1$,  and can obtain simpler and faster algorithms. E.g., in the same model, Chung, Pettie, and Su obtained a randomized algorithm that runs in time $O(\log_{epd^2} n)$ whenever the \emph{criterion} $epd^2<1$ holds~\cite{SuLLL2017}. 
  
\noindent  \textit{LCLs through LLL:} The main recent interest in distributed LLL are constant degree graphs motivated by the study of LCLs. LCL problems are defined on constant degree graphs and in an \emph{LCL problem} each node is given an input from a constant sized set of labels and must output a label from a constant sized set of labels. The problem is characterized by a set of feasible constant-radius neighborhoods (for a formal definition see \Cref{sec:LLL}). Many classic problems are LCL problems, e.g., the problem of finding a vertex coloring with $\Delta$ colors in a graph with maximum degree $\Delta$.
  The systematic study of LCLs in the \LOCAL model initiated by Naor \& Stockmeyer \cite{naor95} has picked up speed over the last years leading to an almost complete classification of the complexity landscape of LCL problems in the \LOCAL model \cite{FOCSchang2016a,CP19,BHKLOS18,BBOS18,RG20}.
  One of the most important results in this line of research is by Chang and Pettie  \cite{CP19} who showed that any $o(\log n)$-round randomized \LOCAL algorithm for an LCL problem $P$ implies the existence of a $T_{\mathsf{LLL}}$-round algorithm for $P$ where $T_{\mathsf{LLL}}$ is the runtime of an LLL algorithm with an (arbitrary) polynomial criterion, e.g., for $p(ed)^{100}$.
  Among other implications, the breakthrough result by Rozho\v{n} and Ghaffari provided an $(\poly\log \log n)$-round LLL  algorithm, in the \LOCAL model, for the case that the dependency degree $d$ is bounded. 
  This implies a gap in the randomized complexity landscape for LCLs.  There is no LCL problem with a complexity strictly between $\poly\log \log n$ and $\Omega(\log n)$ in the \LOCAL model. 
  
  \subsection{Our Results on the Distributed \lovasz Local Lemma and LCLs} 
  Motivated by the progress in \LOCAL using the \emph{LCLs through LLL} application, we aim to design \CONGEST algorithms for the LLL problem. The first observation is that one wants to limit the range of the variables, as any reasonable algorithm should be able to send the value of a variable in a \CONGEST message,  and one also wants to limit the number of variables at a node such that nodes can learn assignments of all variables associated with their bad event(s)  efficiently. We call an LLL instance with dependency degree $d$ \emph{range bounded} if each event only depends on $\poly d$ variables and each variable uses values from a range of size $\poly d$. At first sight, this seems like a very restrictive setting, but in fact most instances of LLL satisfy these requirements. Our main result is contained in the following theorem and meets the $\poly\log \log n$ round  state of the art in the \LOCAL model \cite{RG20}.

 \begin{restatable}{theorem}{theoremBoundedLLL}
 	\label{thm:rangeBoundedLLL}
 	There is a randomized \CONGEST algorithm that with high probability solves any range bounded \lovasz Local Lemma instance that has at most $n$ bad events, dependency degree $d=O(1)$ and satisfies the LLL criterion $p(ed)^8<1$ where $p$ is an upper bound on the probability that a bad event occurs,  in $\poly\log\log n$ rounds.
 \end{restatable}
On the negative side, it is know that a double logarithmic dependency cannot be avoided.
There is an $\Omega(\log\log n)$-round lower bound for solving LLL instances with randomized algorithms in the \LOCAL model  \cite{LLL_lowerbound}, that holds even for constant degree graphs,  and it naturally applies to the \CONGEST model as well. It stems from a lower bound on the aforementioned sinkless orientation problem.  
 
\paragraph*{Implications for the Theory of LCLs.}
As the reduction from \cite{CP19},  that reduces sublogarithmic time randomized algorithms to LLL algorithms, immediately works in \CONGEST, our fast LLL algorithm implies a gap in the complexity landscape of LCLs.
A more precise definition of LCLs and a proof for the following corollary will be presented in the end of \Cref{sec:LLL}.
 
 \begin{restatable}{corollary}{corLCLgap}
 	\label{cor:LCLgap}
 	There is no LCL problem with randomized complexity strictly between $\poly\log\log n$ and $\Omega(\log n)$ in the \CONGEST model. 
 \end{restatable}
In fact, \Cref{cor:LCLgap} together with the results of \cite{Balliu2021}  implies that the \emph{known world} of complexity classes in the sublogarithmic regime in the \CONGEST and \LOCAL model are almost identical. In fact, a difference can only appear in the extremely small complexity regime between $\Omega(\log\logstar n)$ and $O(\logstar n)$ and in the important regime of complexities that lie between $\Omega(\log \log n)$ and $\poly\log \log n$  where complexity classes are not  understood in the \LOCAL model. An immediate implication of \Cref{cor:LCLgap} is a $\poly\log\log n$-round randomized \CONGEST algorithm for $\Delta$-coloring when $\Delta$ is constant, a result that was previously only known in the \LOCAL model \cite{GHKM18}.
Chang and  Pettie \cite{CP19} conjecture that the runtime of LLL in the \LOCAL model is $O(\log\log n)$ on general bounded degree graphs which would further simplify the complexity landscape for LCLs. On trees this complexity can be obtained \cite{CHLPU18} in the \LOCAL model and for the specific LLL instances that arise in the study of LCLs on trees an $O(\log \log n)$-round algorithm has also been found in the \CONGEST model \cite{Balliu2021}.

\subsection{Technical Overview and Background on our Methods} 

Our core technical contribution to obtain \Cref{thm:rangeBoundedLLL} is a bandwidth efficient derandomization of the LLL algorithm by Chung, Pettie, and Su \cite{SuLLL2017} that we combine with the shattering framework of Fischer and Ghaffari \cite{FGLLL17}. To explain these ingredients and how we slightly advance our understanding of network decompositions in the \CONGEST model on the way, we begin with explaining the background on distributed derandomization and network decompositions and the relation to results in the \LOCAL model. 

\paragraph*{Background on Network Decompositions and Distributed Derandomization.} Network decompositions are powerful tools with a range of applications in the area of distributed graph algorithms and were introduced by Awerbuch, Luby, Goldberg, and Plotkin~\cite{awerbuch89}. A $(C,D)$-\emph{network decomposition} (ND) is a partition of the vertices of a graph into $C$ collections of clusters (or color classes of clusters) such that each cluster has diameter\footnote{For the sake of this exposition it is enough to assume that a cluster $\mathcal{C}$ has diameter $D$, that is, any two vertices in the cluster are connected with a path within the cluster of length at most $D$. Actually, often these short paths are allowed to leave the cluster which may cause congestion when one uses these paths for communication in two clusters in parallel. The details of the standard way to model this congestion and its impact are discussed in \Cref{sec:graphDecompositions}.} at most $D$.   Further,  it is required that the clusters in the same collection are \emph{independent}, i.e., are not connected by an edge. This is extremely helpful in the \LOCAL model, e.g., to compute a $(\Delta+1)$-coloring of the network graph one can iterate through the $C$ collections, and within each collection process all clusters in parallel (due to their independence). Due the unbounded message size dealing with a single cluster is trivial. A cluster leader can learn all information about the cluster in time that is proportional to the cluster diameter $D$, solve the problem locally and disseminate the solution to the vertices of the cluster. Thus the runtime scales as $O(C\cdot D)$.  Hence, the objective has been to compute such decompositions as fast as possible and with $C$ and $D$ as small as possible, optimally, all values should be polylogarithmic in $n$. 
 Awerbuch et al.  gave a deterministic \LOCAL algorithm with round complexity and $C$ and $D$ equal to $2^{O(\sqrt{\log n\log\log n})}\gg \poly\log n$. Panconesi and Srinivasan improved both parameters and the runtime to $2^{O(\sqrt{\log n})}$ \cite{panconesi95}. Linial and Saks showed that the optimal trade-off between diameter and the number of colors is $C=D=O(\log n)$  and they provided an $O(\log^2 n)$-round randomized algorithm to compute such decompositions \cite{LS93} . 
 These algorithms remained the state of the art for almost three decades even though the need for an efficient deterministic algorithm for the problem has been highlighted in many papers, e.g.,  \cite{barenboimelkin_book,GKM17,GHK18,CP19,GK19}. 
 
 A few years ago, Ghaffari, Kuhn, and Maus \cite{SLOCAL17} and Ghaffari, Harris, and Kuhn \cite{newHypergraphMatching} highlighted the importance of network decompositions by showing that an efficient deterministic \LOCAL algorithm for network decompositions with $C=D=\poly\log n$ would immediately show that an efficient randomized \LOCAL algorithm for any efficiently verifiable graph problem would yield an efficient deterministic algorithm. Here, all occurrences of \emph{efficient} mean polylogarithmic in the number of nodes of the network. 
 Then, in the aforementioned breakthrough Rozho\v{n} and Ghaffari \cite{RG20} devised such an efficient deterministic algorithm for network decompositions, yielding efficient deterministic algorithms for many problems. The result also had an immediate impact on randomized algorithms.  Many randomized algorithms in the area use the \emph{shattering technique}  that at least goes back to Beck \cite{beck1991LLL} who used the technique to solve LLL instances in centralized settings. It has been introduced to distributed computing by Barenboim, Elkin, Pettie and Schneider in \cite{BEPSv3}.  The shattering technique usually implements the following schematic:
 First, the nodes use  a randomized process and the guarantee is that with high probability \emph{almost} all nodes find a satisfactory output.
 The remainder of the graph is \emph{shattered} into \emph{small components}, that is, after this so called \emph{pre-shattering} phase the unsolved parts of the graph induce small---think of $N=\poly\log n$ size--- connected components.  
 In the \emph{post-shattering phase} one wishes to use an efficient deterministic algorithm, e.g., a deterministic algorithm with complexity $T(n)=\poly\log n$ applied to each small component results in a complexity of $T(N)=T(\poly\log \log n)$. Combining the shattering technique, e.g., \cite{BEPSv3,FGLLL17,ghaffari16_MIS}, the network decomposition algorithm from \cite{RG20} and the derandomization from \cite{GKM17,GHK18} the randomized complexities for many graph problems in the \LOCAL model are $\poly\log \log n$. 
 
\paragraph*{The Challenges in the \CONGEST Model.}
The holy grail would be to obtain a similar derandomization result in the \CONGEST model. But,  even though the network decomposition algorithm from \cite{RG20} immediately works in \CONGEST , we are probably far from  obtaining such a result. Even for simple problems like computing a maximal independent set or a $(\Delta+1)$-coloring one has to work much harder to even get $\poly\log n$ round deterministic algorithms in the \CONGEST model, even if a $(\log n, \log n)$-network decomposition is given for free \cite{CPS20,BKM20}. 
For LLL obtaining such a bandwidth efficient algorithm seems much harder. 
Even in the \LOCAL model one either has to go through the aforementioned derandomization result or one can solve an LLL instance with criterion $p(ed)^{\lambda}$, e.g., think of $\lambda=10$,  if the network decomposition only has $\lambda$ color classes. 
This immediately implies that the cluster diameter is $\Omega(n^{1/\lambda})$ and algorithms that rely on aggregating all information about a cluster in a cluster leader must  use at least $\Omega(n^{1/\lambda})$ rounds~\cite{LS93}. 
In the \LOCAL model there are black box reductions \cite{GKM17,RG20,BEGav18} to compute such decompositions. We provide an analysis of the algorithm by Rozho\v{n} and Ghaffari (and of another algorithm by Ghaffari, Grunau and Rozho\v{n} \cite{GGR20}) where we carefully study the trade-off between number of colors, the diameter of the clusters, and the runtime of the algorithm to obtain such decompositions in \CONGEST.
\medskip

\noindent\textbf{\Cref{thm:constantColors}} (informal)\textbf{.}
\emph{For $k\geq 1$ and $1\leq \lambda <\log n$ there is an $O(k\cdot n^{1/\lambda}\poly\log n)$-round deterministic \CONGEST algorithm to compute a $(\lambda, (k\cdot n^{1/\lambda}\poly\log n)$-network decomposition such that any two clusters with the same color have distance strictly more than $k$. The algorithm works with a mild dependence on the ID space\footnote{To be precise, the dependency on an ID space $\mathcal{S}$ space is a $\logstar|\mathcal{S}| $ factor.}.}

\medskip

For $k=1$, decompositions with few colors similar to \Cref{thm:constantColors} could already be obtained (using randomization) through the early works by Linial and Saks, who also showed that their trade-off of the number of colors and the cluster diameter in \Cref{thm:constantColors} is nearly optimal \cite{LS93}. 
In fact, \Cref{thm:constantColors} does not just provide an improved analysis of the algorithm of \cite{RG20} but it also comes with a mild dependence on the identifier space and an arbitrary parameter to increase the distance between clusters. Both ingredients are also present in our second result on network decompositions (\Cref{thm:mainCongest}) where they are crucial for the proof of \Cref{thm:rangeBoundedLLL}.

\subparagraph*{Our Solution: Range bounded LLLs in \CONGEST.}
We use the shattering framework for LLLs of \cite{FGLLL17} whose pre-shattering phase, as we show, works in the \CONGEST model for range bounded LLLs. As a result we obtain small remaining components of size $N\ll n$, in fact, we obtain $N=O(\log n)$. Furthermore, one can show that the remaining problem that we need to solve on the small components is also an LLL problem. We solve these small instances via a bandwidth efficient derandomization of the LLL algorithm by Chung, Pettie, Su \cite{SuLLL2017}.  Note that we cannot solve the small components without derandomizing their algorithm, as running their randomized algorithm on each component for $T(N)=O(\log N)=O(\log\log n)$ rounds would imply an error probability of $1/N$ which is exponentially larger than the desired high probability guarantee of $1/n$. Thus, we desire to find \emph{good random bits} for all nodes with which we can execute the algorithm from \cite{SuLLL2017} without any error. The goal is to use a network decomposition algorithm to partition the small components into $O(\log N)=O(\log\log n)$ collections of independent clusters. Then, we iterate through the collections and want to obtain good random bits for the vertices inside each cluster. Since the randomized runtime of \cite{SuLLL2017} on a small component would be $T(N)=O(\log N)=O(\log \log n)$, we observe that the random bits of a node $v$ cannot influence the correctness at a node $u$ if $u$ and $v$ are much further  than $T(N)$ hops apart. Thus, similar to \Cref{thm:constantColors} we devise the following theorem to compute network decompositions with large distances between the clusters.  In fact, if we apply the theorem with  $k>2\cdot T(N)$, we obtain independent clusters in each color class of the decomposition. 

\medskip

\noindent\textbf{\Cref{thm:mainCongest}} (informal)\textbf{.}
\emph{For $k\geq 1$  there is an $O(k\cdot \poly\log n)$-round deterministic \CONGEST algorithm to compute a $(\log n, \poly\log n)$-network decomposition  such that any two clusters with the same color have distance strictly more than $k$. The algorithm works with a mild dependence on the ID space.}

\medskip

While \cite{RG20} provided a modification of their algorithm that provides a larger distance between clusters its runtime and cluster diameter depend polylogarithmically not only on the number of nodes in the network but also on the ID space. In \cite{GGR20} Ghaffari, Grunau and Rozho\v{n} have reduced the ID space  in the special case in which the cluster distance $k$ equals one. However, our bandwidth efficient LLL algorithm requires $k\gg 1$ and identifier independence at the same time.  Thus, one can either say we add the ID space independence to \cite{RG20} or we extend the results of \cite{GGR20} to $k>1$.

We already explained why we require that we obtain a network decomposition with a large cluster distance to derandomize an algorithm. The mild dependence on the identifier space in \Cref{thm:mainCongest} is also essential as the small components with $N$ nodes on which we want to use the network decomposition algorithm, actually live in the original communication network $G$. Thus they are equipped with identifiers that are polynomial in $n$, that is, exponential in $N$. 
The fact that the small components live in the large graph makes our life harder when computing a network decomposition but it helps us when designing efficient \CONGEST algorithms. The bandwidth when executing an algorithm on the small components is still $O(\log n)$ bits per edge per round while the components are of size $N\ll n$. Ignoring details for the sake of this exposition,  we use the increased bandwidth to gather enough information about a cluster in a single cluster leader such that it can select good random bits for all nodes of the cluster and in parallel with all other clusters of the same color class due to the large distance between clusters.

We emphasize that there have been several other approaches to derandomize algorithms in the \CONGEST model and discussing all of them here would be well beyond the scope of this work. However, we still believe that our derandomization technique in the post-shattering phase might be of independent interest as it applies to a more general class of algorithms than just the aforementioned LLL algorithm from \cite{SuLLL2017}.

\noindent \textit{Other Implications.} As the runtime of \Cref{thm:mainCongest} is a $\log n$ factor faster than the result in \cite{RG20}, it also improves the complexity of deterministic distance-$2$ coloring with $(1+\eps)\Delta^2$ colors in the \CONGEST model to $O(\log^7 n)$ rounds when plugged into the framework of \cite{HKM20}.
Similar improvements carry over to the approximation of minimum dominating sets \cite{DKM19} and spanner computations \cite{GK18}.
 Due to the identifier independence and possibility to increase the distance between clusters   \Cref{thm:mainCongest}  yields an improved randomized complexity of distance-$2$ coloring. Using \Cref{thm:mainCongest} in the shattering based algorithm from \cite{HKMN20} improves the runtime for $\Delta^2+1$ colors from $2^{O(\sqrt{\log\log n})}$ to $O(\log^7\log n)$ rounds.

\paragraph*{Roadmap:}
In \Cref{ssec:relatedWork}, we provide pointers for further related work, mainly on network decompositions. In \Cref{sec:definitions} we define the models and introduce notation. In \Cref{sec:graphDecompositions} we formally state the result on network decompositions and indicate the main changes to prior work. The formal proofs of these results appear in \Cref{app:redblue,app:RG,app:levels}.  The main part of the paper deals with proving \Cref{thm:rangeBoundedLLL} and appears in \Cref{sec:LLL}.

\subsection{Further Related Work}
\label{ssec:relatedWork}
We already mentioned that \cite{GGR20} provides an efficient deterministic \CONGEST algorithm with a mild dependence on the ID space. More precisely, they provided a $(O(\log n, \log^2 n))$-network decomposition in $O(\log^5 n+\log^4 n\cdot \logstar b)$ rounds, where $\logstar b$ is the dependency on the identifier length $b$~\cite{GGR20}. One drawback that we have ignored until now--- and that also applies to all of our results---is that these network decompositions only have so called \emph{weak diameter}, that is, the diameter of a cluster is only guaranteed to be small if it is seen as a subset of the communication network $G$, that is, the distance between two vertices is measured in $G$ and not only in the subgraph induced by a cluster. In contrast, in the so called \emph{strong network decompositions} the diameter in the subgraph of $G$ that is induced by each cluster has to be small. Very recently, at the cost of increasing the polylogarithmic factors in the runtime the results of \cite{GGR20} were extended to obtain strong $(O(\log n, \log^2 n))$-network decompositions \cite{CG21}. Earlier works by Elkin and Neiman provided an $O(\log^2 n)$ randomized algorithm for computing strong $(O(\log n), O(\log n))$-decompositions \cite{elkin16_decomp}.  

Most previous works in the \CONGEST model do not ensure large distances between clusters. Besides the result in \cite{RG20} that we have already discussed there are two notable exceptions. Ghaffari and Kuhn matched the complexity of the \LOCAL model algorithm by Awerbuch et al. by giving a deterministic $k \cdot 2^{O(\sqrt{\log n \log \log n})}$ round algorithm \cite{GK18}. Later, this was improved  to $k\cdot 2^{O(\sqrt{\log n})}$  rounds  by Ghaffari and Portmann \cite{GP19}. In both cases, the decomposition parameters were identical to the runtimes. None of these results is the state of the art anymore, except for the fact that they compute strong network decompositions with large distances between the clusters.


Barenboim, Elkin, and Gavoille provide various algorithms with different trade-offs between the number of colors and the cluster diameter, most notably a randomized algorithm to compute a strong network decomposition with diameter $O(1)$ and $O(n^{\eps})$ colors \cite{BEGav18}. A similar result with $O(n^{1/2+\eps})$ colors was obtained by Barenboim in \cite{B12}.

Brandt, Maus, and Uitto and Brandt, Grunau, and Rozho\v{n} have shown that LLL instances with an exponential LLL criterion, that is, $p2^d<1$ holds, can be solved in $O(\logstar n)$ rounds on bounded degree graphs \cite{BMU19,BGR20}.  Furthermore, it is known that $\Omega(\logstar n)$ rounds cannot be beaten under any LLL criterion that is a function of the dependency degree $d$ \cite{SuLLL2017}. 
For an exponential criterion that satisfies $p2^d \geq 1$, it follows from the works of Brandt et al.~and Chang, Kopelowitz, and Pettie~\cite{FOCSchang2016a} that there is a lower bound of $\Omega(\log n)$ rounds for deterministic LLL. Hence \cite{BMU19,BGR20} and \cite{FOCSchang2016a} show that there is a sharp transition of the distributed complexity of LLL at $2p^d=1$. 
Before the result in \cite{RG20} improved the runtime to $\poly\log\log n$,  Ghaffari, Harris, and Kuhn gave the state of the art randomized LLL algorithm in the \LOCAL model. On constant degree graphs its runtime was described by a tower function and lies between $\poly\log\log n$ and $2^{O(\sqrt{\log\log n})}$~  \cite{GHK18}. We are not aware of any works that explicitly studied LLL in the \CONGEST model before our paper.


\section{Models, LCLs \& Notation}
\label{sec:definitions}

Given a graph $G=(V,E)$ the hop distance in $G$ between two vertices $u,v\in V$ is denoted by $dist_G(u,v)$. 
For a vertex $v\in V$ and a subset $S\subseteq V$ we define $dist_G(v,S)=\min_{u\in S} \{dist(v,u)\}\in \mathbb{N}\cup \{\infty\}$. For two subsets $S,T\subseteq V$ we define $dist_G(S,T)=\min_{v\in T}dist_G(v,S)$.  For an integer $n$ we denote $[n]=\{0,\ldots,n-1\}$.

\subparagraph*{The \LOCAL and \CONGEST Model of distributed computing~\cite{linial92,peleg00}.} In both models the graph is abstracted as an $n$-node network $G=(V, E)$ with maximum degree at most $\Delta$. Communication happens in synchronous rounds. Per round, each node can send one message to each of its neighbors. At the end, each node has to know its own part of the output, e.g., its own color. In the \LOCAL model there is no bound on the message size and in the \CONGEST model messages can contain at most $O(\log n)$ bits. Usually, in both models nodes are equipped with $O(\log n)$ bit IDs (polynomial ID space) and initially, nodes know their own ID or their own color in an input coloring but are unaware of the IDs of their neighbors. Randomized algorithms do not use IDs (but they can create them from space $[n^c]$ for a constant $c$ and with a $1/n^c$ additive increase in the error probability), have a fixed runtime and need to be correct with probability strictly more than $1-1/n$ (Monte Carlo algorithm). Actually, algorithms are defined such that a parameter $n$ is provided to them, where $n$
represents an upper bound on the number of nodes of the graph and the algorithm's guarantees have to hold on any graph with at most $n$ nodes. 
For technical reasons in our gap results, we assume that the number of random bits used by each node is bounded by some finite function $h(n)$. We note that the function $h(n)$ can grow arbitrarily fast and that this assumption is made in previous works as well~\cite{FOCSchang2016a, Balliu2021}.

\medskip

\noindent \textbf{LCL definition \cite{naor95}.} An \emph{LCL problem $\Pi$} is a tuple $(\Sigma_{\mathrm{in}},\Sigma_{\mathrm{out}},F,r)$ satisfying the following.
\begin{itemize}
	\item Both $\Sigma_{\mathrm{in}}$ and $\Sigma_{\mathrm{out}}$ are constant-size sets of labels,
	\item the parameter $r$ is an arbitrary constant, called the \emph{checkability radius} of $\Pi$,
	\item $F$ is a finite set of pairs $(H = (V^H,E^H),v)$, where:
	\begin{itemize}
		\item $H$ is a graph, $v$ is a node of $H$, and the radius of $v$ in $H$ is at most $r$;
		\item Every pair $(v,e) \in V^H \times E^H$ is labeled with a label in  $\Sigma_{\mathrm{in}}$ and a label in $\Sigma_{\mathrm{out}}$.
	\end{itemize}
\end{itemize}
In an instance of an LCL problem $\Pi=(\Sigma_{\mathrm{in}},\Sigma_{\mathrm{out}},F,r)$ on a graph $G=(V,E)$ each node receives an \emph{input label} from $\Sigma_{\mathrm{in}}$. 
An algorithm solves $\Pi$ if, for each node $v\in V$, the radius-$r$ hop neighborhood of $v$ together with the input labels and computed outputs for the nodes lies in the set of \emph{feasible labelings $F$}. 



\section{Graph Decompositions}
\label{sec:graphDecompositions}
In this section we state our results on network decompositions and provide the respective definitions. At the end of the section we provide a short overview of the techniques that we use to prove the results. The formal proofs appear in \Cref{app:redblue,app:RG,app:levels}.

\subsection{Cluster Collections and Network Decompositions with Congestion}

Classically, one defines a \emph{(weak) $(c,d)$-network decomposition} as a coloring of the vertices of a graph $G$ with $c$ colors such that the connected components of each color class have weak diameter at most $d$. The  \emph{weak diameter} of a subset $\mathcal{C}\subseteq V$ is the maximum distance in $G$ that any two vertices of $\mathcal{C}$ have.   

We use the following definition that augments one color class of a network decomposition with a communication backbone for each cluster. A \emph{cluster $\mathcal{C}$} of a graph is a subset of nodes. 
A \emph{Steiner tree} is a rooted tree with nodes labeled as \emph{terminal} and \emph{nonterminal}. 
\begin{definition}[cluster collection]
	A collection of clusters of a graph $G=(V,E)$ consists of subsets of vertices $\mathcal{C}_1,\dots,\mathcal{C}_p\subseteq V$  and has Steiner radius $\beta$ and unique $b$-bit cluster identifiers if it comes with associated Steiner subtrees $T_1,\ldots,T_p$ of $G$ such that clusters are disjoint, i.e., $\mathcal{C}_i\cap \mathcal{C}_j=\emptyset$ for $i\neq j$ and for each $i \in 1, 2, \ldots, p$ we have
	\begin{enumerate}
		\item cluster $\mathcal{C}_i$ has a unique $b$-bit identifier $id_{\mathcal{C}_i}$, 
		\item the terminal nodes of  Steiner tree $T_i$ of cluster $\mathcal{C}_i$ are formed by $\mathcal{C}_i$ ($T_i$ might contain nodes $\notin \mathcal{C}_i$ as non terminal nodes),
		\item the edges of Steiner Tree $T_i$ are oriented towards a \emph{cluster leader} $\ell_{\mathcal{C}_i}$,
		\item Steiner tree $T_i$ has diameter at most $\beta$ (Steiner tree diameter). 
	\end{enumerate}
	The collection has \emph{congestion} $\kappa$ if each edge in $E$ is contained in at most $\kappa$ Steiner trees. 	
	When we compute a cluster collection in the \CONGEST model, we require that each node of a cluster knows the cluster identifier.
	Furthermore, we require that for each edge of $T_i$ that is incident to some $v\in V$, node $v$ knows $id_{\mathcal{C}_i}$ and additionally $v$ knows the direction of the edge towards the root $r_{\mathcal{C}_i}$.
\end{definition}
We now define our notion of a network decomposition with a communication backbone.

\begin{definition}[network decomposition with congestion]
	\label{def:nd}
	Let $k\geq 1$. 	A \emph{$(C,\beta)$-network decomposition} with cluster distance $k$ and \emph{congestion} $\kappa$ of a graph $G=(V,E)$ is a partition of $V$ into $C$ cluster collections with Steiner radius $\beta$ and congestion $\kappa$ such that any two clusters $\mathcal{C}\neq \mathcal{C}'$ in the same cluster collection have distance strictly more than $k$, i.e., $dist_G(\mathcal{C},\mathcal{C}')>k$. 
	
	The  $C$ cluster collections are also called the \emph{$C$ color classes} of the decomposition.
\end{definition}
Instead of a network decomposition with cluster distance $k$, we often just speak of a network decomposition of $G^k$. In all cases the Steiner trees are given by vertices and edges in $G$ and not by edges in $G^k$, which would corresponds to paths in $G$. The default value for $k$ is $1$ which we use whenever we do not specify its value. in this case the definition corresponds with the classic network decomposition.

\subsection{Network Decomposition and Ball Carving Algorithms}
At the core of our network decomposition algorithm is the following \emph{ball carving} result to form one color class of the decomposition, i.e., to form one cluster collection. The name ball carving stems from the original existential proof for network decompositions \cite{awerbuch89}. 
\begin{restatable}[ball carving, $k\geq 1$]{lemma}{theoremSlowBallCarving}
	\label{thm:ballCarvingSlow}
	Let $k\geq 1$ be an integer and $x \geq 1$ (both potentially functions of $n$). Then there is a deterministic distributed \CONGEST (bandwidth $b$ and $b$-bit identifiers) algorithm that, given a graph $G=(V,E)$ with at most $n$ nodes and a subset $S\subseteq V$, computes a cluster collection $\mathcal{C}_1,\ldots, \mathcal{C}_p\subseteq S$ with
	\begin{itemize}
		\item $|\mathcal{C}_1\cup \ldots \cup \mathcal{C}_p|\geq (1 - 1/x) \cdot |S|$,
		\item Steiner radius $\beta=O(k \cdot x\cdot\log^3 n)$ and congestion $\kappa=O(\log n \cdot \min\{k,x\cdot \log^2 n \})$,
		\item the pairwise distance in $G$ between $\mathcal{C}_i$ and $\mathcal{C}_j$ for $1\leq i\neq j\leq p$ is strictly more than $k$. 
	\end{itemize}
	The runtime  is $O(k\cdot x\cdot \log^4 n\cdot \logstar b) + O(k\cdot x^2\cdot \log^6 n\cdot \min\{k,x\cdot \log^2  n\})$ rounds. 
\end{restatable}
By using the aggregation tools for overlapping Steiner trees of \cite{GGR20} \Cref{thm:ballCarvingSlow}  is by a $\log n$ factor faster than the corresponding result in \cite{RG20}. 
The ball carving result immediately imply the following network decomposition results. 
\begin{theorem}[$O(\log n)$ colors]
	\label{thm:mainCongest}
	For any (potentially non constant) $k\geq1$, there is a deterministic \CONGEST algorithm with bandwidth $b$ that, given a graph $G$ with at most $n$ nodes and unique $b$-bit IDs from an exponential ID space, computes  a (weak) $(O(\log n), O(k \cdot \log^3 n))$-network decomposition with cluster distance $k$ and with congestion $O(\log^2 n \cdot \min \{ k, \log^2 n \})$ in $O(k\cdot \log^7 n\cdot \min\{k,\log^2 n\})$ rounds. 
\end{theorem}
\begin{proof}
	Set $x=2$ and apply \Cref{thm:ballCarvingSlow} for $\log_2 n$ iterations, always with the set of nodes that have not been added to a cluster in one of the previous iterations.
Notice that the bandwidth demand by \Cref{thm:ballCarvingSlow} is satisfied since the IDs have at most $b$ bits.
Each iteration forms a separate color class and as the number of remaining vertices is reduced by a factor $2$ in each iteration each vertex is contained in one cluster at the end. 
\end{proof}

We also analyze a more involved ball carving algorithm of \cite{GGR20} for different parameters in order to obtain decompositions with fewer colors.  Then \Cref{thm:ballCarvingSlow,thm:ballCarvingFast} imply the following result for network decompositions with few colors. 

\begin{restatable}[faster ball carving, $k=1$]{lemma}{theoremFastBallCarving}
	\label{thm:ballCarvingFast}
	Let $x \geq 1$ (potentially a function of $n$). 	Then there is a deterministic distributed \CONGEST (bandwidth $b$ and $b$-bit identifiers) algorithm that, given a graph $G=(V,E)$ with at most $n$ nodes and a subset $S\subseteq V$, computes a cluster collection $\mathcal{C}_1,\ldots, \mathcal{C}_p\subseteq S$ with
	\begin{itemize}
		\item $|\mathcal{C}_1\cup \ldots \cup \mathcal{C}_p|\geq (1 - 1/x) \cdot |S|$,
		\item Steiner radius $\beta=O(x\cdot\log^2 n)$ and congestion $\kappa=O(\log n)$,
		\item the pairwise distance in $G$ between $\mathcal{C}_i$ and $\mathcal{C}_j$ for $1\leq i\neq j\leq p$ is strictly more than $1$. 
	\end{itemize}
	The runtime of the algorithm is	$O(x^2 \log^4n)$ rounds.
\end{restatable}

\begin{theorem}[few colors]	
	\label{thm:constantColors}
	For any  $\lambda\leq \log n$ there is a deterministic \CONGEST algorithm with bandwidth $b$ that, given a graph $G$ with at most $n$ nodes and unique $b$-bit IDs from an exponential ID space, computes  a weak $(\lambda, n^{1/\lambda} \log^2 n)$-network decomposition of $G$ in $O(\lambda \cdot n^{2/\lambda}\cdot \log^4 n)$ rounds with congestion $\kappa=O(\log n)$.  
	
	For any (possibly non constant) $k\geq 1$ and any  $\lambda\leq \log n$ there is a deterministic \CONGEST algorithm that, given a graph $G$ with at most $n$ nodes and unique IDs from an exponential ID space, computes  a weak $(\lambda, k\cdot n^{1/\lambda} \log^3 n)$-network decomposition of $G^k$ in $O(k\cdot n^{2/\lambda}\cdot \log^6 n \cdot \min\{k,x\log^2 n\})$ rounds.  
\end{theorem}
\begin{proof}
For the first result apply \Cref{thm:ballCarvingFast} with $x=n^{1/\lambda}$ and for $\lambda$ iterations, always with the set of nodes that have not been added to a cluster in one of the previous iterations.
Notice that the bandwidth demand by \Cref{thm:ballCarvingFast} is satisfied since the IDs have at most $b$ bits. 
For the second result do the same with \Cref{thm:ballCarvingSlow}. In both cases the clusters created in each iteration form a separate color class. The number of remaining vertices decreases by a factor $x$ in each iteration and after $\lambda$ iterations all vertices have been clustered. 
\end{proof}

\Cref{thm:mainCongest,thm:constantColors} also work with more general ID spaces, as long as the IDs fit into a \CONGEST message. To simplify the statements and because exponential IDs are (currently)  the main application we do not explicitly state their mild dependence on the size of the ID space $\mathcal{S}$ in theorems. It can be quantified by $O(x\cdot \poly\log n \cdot \logstar |\mathcal{S}|)$ where the polylogarithmic terms are strictly dominated by the ones in the theorems.

\subparagraph*{Summary of Our Ball Carving Contributions.} Both, the ball carving results in \cite{RG20} and \cite{GGR20} begin with each vertex of the to be partitioned set $S$ (see \Cref{thm:ballCarvingSlow,thm:ballCarvingFast}) being its own cluster inheriting its node ID as the cluster ID. Then they implement a distributed \emph{ball growing} approach in which vertices leave their cluster, join other clusters, or become dead, i.e., are disregarded; even whole clusters can dissolve, e.g., if all their vertices join a different cluster. The goal is to ensure that at the end no two alive clusters are neighboring, that is, in distance $k$, that the fraction of vertices of $S$ that are declared dead is small and that the diameter of the clusters does not become too large. In fact, as vertices can leave clusters one only obtains guarantees on the weak diameter of clusters. To this end the algorithm of \cite{RG20} iterates through the bits of the cluster IDs and in phase $i$, vertices change their clusters or become dead such that at the end of the $i$-th phase no two neighboring clusters have the same $i$-th bit in their cluster ID. If cluster IDs have $b$ bits, then after the $b$-th phase each remaining cluster is not connected to any other cluster. Thus, the runtime crucially depends on the size of the cluster IDs. The main change in \cite{GGR20} to remove this ID space dependence is to replace the bits ($0$ or $1$) used in phase $i$ with a coloring of the clusters with two colors, red and blue. A new coloring is computed in every phase and the crucial property of the coloring that ensures progress towards the separation of the clusters is that each connected component of clusters has roughly the same number of red and blue clusters. Thus, intuitively,  at the end of a phase red and blue clusters are separated and the sizes of all connected components decrease. 

In our algorithms we follow the same high level approach but extend the techniques (for identifier ``independence") to work with a larger separation between the clusters, and also for the computation of decompositions with fewer colors. 
A crucial difficulty is that we cannot quickly disregard all dead vertices and reason that each connected components of clusters in $G$ can be treated independently. 
It might even be that two clusters are connected by a path with $\leq k$ hops and the path might contain dead vertices. Instead, we consider connected components in $G^k$ and always keep all vertices of $G$ in mind. As two clusters that are adjacent in such a component might not be adjacent in the original graph we ensure that no congestion appears \emph{in between} the clusters. 
Careful algorithm design and reasoning is needed when computing a balanced coloring of clusters in such components. 

To obtain the network decompositions with fewer colors, e.g., with $\lambda=10$ colors,  we analyze our algorithm for a suitable choice of parameters and show that all congestion parameters are not affected by the choice of $\lambda$. 


\section{Distributed Range Bounded LLL and its Implications}
\label{sec:LLL}
The objective of this section is to prove the following theorem. 

\theoremBoundedLLL*


\subparagraph*{Distributed \lovasz Local Lemma (LLL).} In a distributed \lovasz Local Lemma instance we are given a set of independent random variables $\mathcal{V}$ and a family $\mathcal{X}$ of (bad) events $\mathcal{E}_1,\ldots,\mathcal{E}_n$  on these variables. Each bad event $\mathcal{E}_i\in \mathcal{X}$ depends on some subset $vbl(\mathcal{E}_i)\subseteq \mathcal{V}$ of the variables. Define the dependency graph $H_{\mathcal{X}}= (\mathcal{X},\{(\mathcal{E},\mathcal{E}') \mid vbl(\mathcal{E}) \cap vbl(\mathcal{E}')\neq \emptyset \})$ that connects any two events
which share at least one variable. Let $\Delta_H$ be the maximum degree in this graph, i.e., each event $\mathcal{E}\in\mathcal{X}$ shares variables with at most $d=\Delta_H$ other events $\mathcal{E}'\in \mathcal{X}$. The \lovasz Local Lemma \cite{LLL73}  states that $Pr(\bigcap_{\mathcal{E}\in \mathcal{X}} \mathcal{E})>0$ if $epd<1$. In the \emph{distributed LLL problem} each bad event and each variable  is assigned to a vertex of the communication network such that all variables that influence a bad event $\mathcal{E}$ that is assigned to a vertex $v$ are either assigned to $v$ or a neighbor of $v$.
The objective is to compute an assignment for the variables such that no bad event occurs.  At the end of the computation each vertex $v$ has to know the values of all variables that influence one of its bad events (by definition these variables are assigned to $v$ or one of its neighbors).  In all  results on distributed LLL in the \LOCAL model that are stated below  the communication network is identical to the dependency graph $H$. Most result in the distributed setting require stronger LLL criteria.

\begin{algorithm} \caption{The algorithm from \cite{SuLLL2017} iteratively resamples all variables of local ID minima of violated events. IDs can be assigned in an adversarial manner. The runtime is with high probability $O(\log_{epd^2} n)$ rounds under LLL criterion $epd^2<1$.} \label{alg:CPSLLL}
	Initialize a random assignment to the variables \\
	Let $\mathcal{F}$ be the set of bad events under the current variable assignment\\
	\While{$\mathcal{F} \neq  \emptyset$}{
		Let $I = \big\{A \in\mathcal{F} \mid  ID(A) = \min\{ID(B) | B \in  N^+_{\mathcal{F}}(A)\}\big\}$ \\
		Resample $vbl(I) =  \cup_{A\in I}vbl(A)$
	}
\end{algorithm}

Despite its simplicity---the algorithm simply iteratively resamples local ID minima of \emph{violated events}, i.e., bad events that hold under the current variable assignment ---Algorithm~\ref{alg:CPSLLL} results in the following theorem. 
\begin{theorem}[\cite{SuLLL2017}, Algorithm~\ref{alg:CPSLLL}]
	Given an LLL instance with condition $epd^2 < 1$, the \LOCAL Algorithm~\ref{alg:CPSLLL} run for  $O(\log_{epd^2} n)$ rounds and has error probability $<1/n$ on any LLL instance with at most $n$ bad events. 
\end{theorem}

Next, we define the notion of a range bounded LLL instances. 

\begin{definition}[Range bounded \lovasz Local Lemma]
	An instance $(\mathcal{V}, \mathcal{X})$ of the \lovasz Local Lemma with bad events $\mathcal{E}_1,\ldots,\mathcal{E}_n$ and dependency graph $H$ is \emph{range bounded} if 
	$|vbl(\mathcal{E}_i)|= \poly \Delta_H$ for each $1\leq i\leq n$ and the value of each random variable $x\in \mathcal{V}$ can be expressed by at most $O(\log \Delta_H)$ bits. 
\end{definition}
Note that most applications of the \lovasz Local Lemma in the distributed form are range bounded instances. Algorithm~\ref{alg:CPSLLL}  can be executed with the same asymptotic runtime for range bounded LLL instances and uses few random bits while doing so.  

\begin{lemma}
	\label{lem:LLLboundedRangeRandomized}
	For a range bounded LLL instance with condition $epd^2<1$ with constant dependency degree there is a randomized \CONGEST (constant bandwidth) algorithm  that runs in $O(\log_{epd^2} n)$ rounds and has error probability $<1/n$ on any (dependency) graph with at most $n$ nodes. Furthermore, in the algorithm each node only requires access to $O(\log n)$ random values from a bounded range. The unique IDs can be replaced by an acyclic orientation of the edges of the graph. 
\end{lemma}
\begin{proof}
	Let $\Delta$ be the degree of the dependency graph. 
	 In Algorithm~\ref{alg:CPSLLL} local minima with regard to an arbitrary ID assignment resample their variables. The algorithm is oblivious to the values of the IDs and the IDs can be replaced with an arbitrary acyclic orientation on the edges. To check which events need to be resampled nodes only need to learn the currently sampled values of adjacent nodes (and the orientation of the edges). As the LLL is range bounded the values of variables fit into $O(1)$ bits and can be communicated in one round. In each iteration of the algorithm, each vertex resamples only its own random variables (or none). To resample one variable it requires a random value from the bounded range $\poly \Delta=O(1)$, and in each iteration it resamples at most $\poly \Delta=O(1)$ variables.
\end{proof}

The currently fastest randomized algorithm for LLL on bounded degree graphs with polynomial criterion in the \LOCAL model is based on two main ingredients.  The first ingredient is a deterministic $\poly\log n$-round LLL algorithm obtained via a derandomization (cf. \cite{GKM17,GHK18}) of the $O(\log^2 n)$-round randomized algorithm \cite{MoserTardos10} using the breakthrough efficient network decomposition algorithm \cite{RG20}. The second ingredient is the (randomized) shattering framework for LLL instances by \cite{FGLLL17}, which \emph{shatters} the graph into small unsolved components of logarithmic size. Then, these components are solved independently and in parallel with the deterministic algorithm. 

We follow the same general approach of shattering the graph via the methods of \cite{FGLLL17}, but need to be more careful to obtain a bandwidth efficient algorithm when dealing with the small components. 

We begin with a lemma that describes how fast we can gather information in a cluster leader with limited bandwidth. Its proof uses standard pipelining techniques (see e.g. \cite[Chapter 3]{peleg00} and fills each $b$-bit messages with as much information as possible.
\begin{lemma}[Token learning]
	\label{lem:learningTokens}
	Let $G$ be a communication graph on $n$ vertices in which each node can send $b$ bits per round on each edge. Assume a cluster collection with Steiner radius $\beta$ and congestion $\kappa$ in which each cluster $\mathcal{C}$ is of size at most $N$ and each vertex holds at most $x$ bits of information. Then, in parallel each cluster leader $\ell_{\mathcal{C}}$ can learn the information of each vertex of $\mathcal{C}$ in $O\big(\kappa\cdot (\beta+N\cdot x/b)\big)$ rounds.
	In the same runtime the leader $\ell_{\mathcal{C}}$ can disseminate $x$ bits of distinct information to each vertex of $\mathcal{C}$. 
\end{lemma}
Next, we prove our core derandomization result that we use to obtain efficient deterministic algorithms for the small components that arise in the post-shattering phase. It might be of independent interest. 
\begin{lemma}[Derandomization]
	\label{lem:derandomization}
	Consider an LCL problem $P$, possibly with promises on the inputs. Assume a $T(n)$-round randomized \CONGEST (bandwidth $b=\Theta(\log n)$) algorithm $\mathcal{A}$ for $P$ with error probability $<1/n$ on any graph with at most $n$ nodes that uses at most $O(\log n)$  random bits per node. 
	
	Then for any graph on at most $N$ nodes there is a deterministic $T(N)\cdot \poly\log N$ round \CONGEST (with bandwidth $b=\Theta(N)$) algorithm $\mathcal{B}$ to solve $P$ under the same promises, even if the ID space is exponential in $N$.
\end{lemma}
\begin{proof}
	For an execution of algorithm $\mathcal{A}$ on a graph with $N$ nodes define for each $v\in V$  an indicator variable $X_v$ that equals to $1$ if the verification of the solution fails at node $v$ and $0$ otherwise. The value of $X_v$ depends on the randomness of nodes in the $(T(N)+r)$ hop neighborhood of $v$, where $r = O(1)$ is defined by the LCL problem. 

We design an efficient algorithm $\mathcal{B}$ that deterministically fixes the (random bits) of each node $v\in V$ such that $X_v=0$ for all nodes. Thus, executing $\mathcal{A}$ with these random bits solves problem $P$. During the execution of $\mathcal{B}$ we fix the random bits of more and more nodes. At some point during the execution of algorithm $\mathcal{B}$ let  $X\subseteq V$ denote the nodes whose randomness is already fixed and let $\phi_X$ be their randomness (we formally define $\phi_X$ later). The crucial invariant that we maintain for each cluster $C$ at all times is the following
\begin{align}
	\label{eqn:invariant}
	\textbf{Invariant: ~} & E\left[\sum_{v\in V}X_v \mid \phi_X\right]<1. 
\end{align}
The invariant will be made formal in the rest of the proof. Initially (when $X=\emptyset$), it holds with the linearity of expectation and as the error probability of the algorithm implies that $E[X_v]=P(X_v=1)<1/n$ holds.  

\paragraph*{Algorithm $\mathcal{B}$:} Compute a weak network decomposition with cluster distance strictly more than $4(T(N)+r)$ with $O(\log N)$ color classes, $\beta=T(N)\cdot \poly\log N$ Steiner tree cluster radius and congestion $\kappa=O(\log N)$. Iterate through the color classes of the decomposition and consider each cluster separately. We describe the process for one cluster $\mathcal{C}$ when processing the $i$-th color class of the network decomposition. We define three layers according to the distance to vertices in $\mathcal{C}$. Denote the vertices in $\mathcal{C}$ as $W_{\mathcal{C}}^0$, the vertices in $V\setminus W_{\mathcal{C}}^0$ in distance at most $r+T(N)$ from $\mathcal{C}$ as $W_{\mathcal{C}}^1$ and the vertices in $V\setminus (W_{\mathcal{C}}^0\cup W_{\mathcal{C}}^1)$ with distance at most $2(r+T(N))$ from $\mathcal{C}$ as $W_{\mathcal{C}}^2$. Define $W_{\mathcal{C}}=W_{\mathcal{C}}^0\cup W_{\mathcal{C}}^1\cup W_{\mathcal{C}}^2$. Note that we have $W_{\mathcal{C}}\cap W_{\mathcal{C}'}=\emptyset$  for two distinct clusters $\mathcal{C}, \mathcal{C}'$ in the same color class of the decomposition due to the cluster distance. Further, even if \Cref{lem:derandomization} is applied to a subgraph $H$ (with size at most $N$) of a communication network $G$ we have $|W_{\mathcal{C}}|\leq N$, as $W_{\mathcal{C}}$ only contains nodes of $H$.

\noindent \textit{Dealing with one cluster $\mathcal{C}$:} Extend the Steiner tree of $\mathcal{C}$ to $W_{\mathcal{C}}$ using a BFS, ties broken arbitrarily. Notice that no two BFS trees interfere with each other since we handle different color classes separately and two clusters of the same color are in large enough distance.
Use the Steiner tree to assign new IDs (unique only within the nodes in $W_{\mathcal{C}})$) from range $[N]$ to all nodes in $W_{\mathcal{C}}$. Different clusters use the same set of IDs and each ID can be represented with $O(\log N)$ bits. Each node of $W_{\mathcal{C}}$ learns  about the new ID of each of its (at most) $\Delta$ neighbors in $W_{\mathcal{C}}$ in one round and uses $O(\Delta\cdot \log N)=O(\log N)$ bits to store its adjacent edges. 
Using \Cref{lem:learningTokens}, the cluster leader $\ell_{\mathcal{C}}$ learns the whole topology, inputs and already determined random bits of $G[W_{\mathcal{C}}]$ in $O(\kappa\cdot (\beta + (N\cdot \poly\log N)/b))$ rounds. 
In this step, the input also includes an acyclic orientation of the edges between vertices of $W_{\mathcal{C}}$, where an edge is oriented from $u\in W_{\mathcal{C}}$ to $v\in W_{\mathcal{C}}$   if the original ID of $v$ is larger than the original ID of $v$. 

For a node $v$ let $R_v$ describe its randomness. When we process cluster $\mathcal{C}$ we determine values  $\{r_v \mid v\in \mathcal{C}\}$ for the random variables $\{R_v \mid v \in \mathcal{C}\}$. Let $X$ be a  set of vertices $v$ for which we have already determined values $r_v$. Then we denote $\phi_X=\bigwedge_{v\in X} (R_v=r_v)$.

\begin{claim}Assume that Invariant~(\ref{eqn:invariant}) holds for $X=\bigcup_{\mathcal{C}\text{ has color $<i$}}\mathcal{C}$ before processing the clusters of color $i$. Then the cluster leaders $\ell_{\mathcal{C}}$ of all clusters $\mathcal{C}$ of color $i$, can in parallel,  find values  $\{r_v \mid v \in \mathcal{C}\}$ such that Invariant~(\ref{eqn:invariant}) holds afterwards for $X'=\bigcup_{\mathcal{C}\text{ has color $\leq i$}}\mathcal{C}$. 
\end{claim} 
	If we considered just a single cluster $\mathcal{C}$, independently from all other clusters with color $i$, then there are choices for the values of $r_v$ satisfying the claim due to the law of total probability, e.g., as used in the method of conditional expectation,  and as $E\big[\sum_{v\in V}X_v\mid \phi_X\big]<1$ holds before we fix any randomness of vertices inside $\mathcal{C}$.
\begin{proof}

  The only information that cluster leaders need to perform the necessary calculations is information on the topology of $G[W_{\mathcal{C}}]$, inputs to nodes in $W_{\mathcal{C}}$, the relative order of adjacent IDs of nodes in $G[W_{\mathcal{C}}]$ and already determined randomness of nodes in $X\cap W_{\mathcal{C}}$. 
  
  We show that the randomness of all vertices in clusters with color $i$ can be fixed with this knowledge.   
  To this end the cluster leader $r_{\mathcal{C}}$ fixes the randomness $\phi_{\mathcal{C}}$ such that 
	\begin{align*}
		E\big[\sum_{v\in W_{\mathcal{C}}}X_v\mid \phi_{\mathcal{C}}\wedge  \phi_{X\cap W_{\mathcal{C}}}\big]\leq E\big[\sum_{v\in W_{\mathcal{C}}}X_v\mid \phi_{X\cap W_{\mathcal{C}}}\big]
	\end{align*} 
holds. Such a choice for $\phi_{\mathcal{C}}$ exists, again due to the law of total probability, and $\ell_{\mathcal{C}}$ has full information to compute such values as all values in the inequality only depend on information that vertices in $W_{\mathcal{C}}$ sent to $\ell_{\mathcal{C}}$. In particular, the randomness of nodes in $\mathcal{C}$ only influences the random variables $X_v$ of nodes in $W_{\mathcal{C}}^0\cup W_{\mathcal{C}}^1$ and computing  this influence only requires knowledge from $W_{\mathcal{C}}^0\cup W_{\mathcal{C}}^1\cup W_{\mathcal{C}}^2$.

After all cluster leaders of clusters with color $i$ fix the randomness of the vertices in their clusters according to the above method the invariant is still satisfied since the randomness $R_v$ for $v\in \mathcal{C}$ does not influence the random variable $X_u$ for $u\in V\setminus W_{\mathcal{C}}$. 
More formally, let $X$ be the set of vertices with fixed randomness before we process the $i$-th color class of clusters and let $Y$ be the set of vertices whose randomness we fix when processing the $i$-th color class. Let $X'=X\cup Y$. Let $W=\bigcup_{\mathcal{C}\text{ has color $i$}}W_{\mathcal{C}}$. Recall that $W_{\mathcal{C}}\cap W_{\mathcal{C}'}=\emptyset $ for $\mathcal{C}\neq\mathcal{C}'$. 
Due to the linearity of expectation and the aforementioned reasons we obtain
\begin{align*}
	E\left[\sum_{v\in V}X_v\mid \phi_{X'}\right] 
	& = E\left[\sum_{v\in V\setminus W}X_v\mid \phi_{X'}\right]+ \sum_{\mathcal{C}\text{ has color $i$}} E\left[\sum_{v\in W_{\mathcal{C}}}X_v\mid \phi_{X'\cap W_{\mathcal{C}}}\right] \\
	& \leq  E\left[\sum_{v\in V\setminus W}X_v\mid \phi_{X}\right]+ \sum_{\mathcal{C}\text{ has color $i$}}E\left[\sum_{v\in W_{\mathcal{C}}}X_v\mid \phi_{X\cap W_{\mathcal{C}}}\right] \\
	& = E\left[\sum_{v\in V}X_v\mid \phi_{X}\right]<1.
\end{align*}
	All clusters can be processed in parallel as their distance is strictly more than $4(T(N)+r)$ and the choices in cluster $\mathcal{C}$ do not change the probability for  $X_v=1$ for  $v\in\mathcal{C}'\neq \mathcal{C}$.
\end{proof}
	Once the cluster leader $\ell_{\mathcal{C}}$ has fixed the randomness for the vertices in $\mathcal{C}$, it disseminates them to all vertices of its cluster via \Cref{lem:learningTokens} and we continue with the next color class of the network decomposition. At the end of the algorithm each node knows the values of its random bits and one can execute $T(N)$ with it, as the initial algorithm $\mathcal{A}$ also works with bandwidth $b$. 
	
	Due to the Invariant~(\ref{eqn:invariant}) we obtain that the distributed random bits are such that	$E[\sum_{v\in V}X_v\mid \phi_V]<1$ at the end of the algorithm, but as all $X_v$'s are in $\{0,1\}$ and there is no randomness involved (each vertex is processed in at least one cluster and thus we fix $R_v=r_v$ for each $v\in V$) we obtain that $X_v=0$ for all $v\in V$, that is, the algorithm does not fail at any vertex.

\noindent \textbf{Runtime:} Computing the network decomposition takes $T(N)\cdot\poly\log N$ rounds via \Cref{thm:mainCongest}. The runtime for processing one color class of the network decomposition is bounded as follows.  Collecting the topology and all other information takes $T(N) \cdot \poly\log N$ rounds due to \Cref{lem:learningTokens}, fixing randomness locally does not require communication. Disseminating the random bits to the vertices of each cluster takes $\poly\log N$ rounds, again due to \Cref{lem:learningTokens}. 
As there are only $\log N$ color classes the runtime for computing good random bits for all nodes is $T(N)\cdot\poly\log N$. Executing algorithm $\mathcal{A}$ with these random bits takes $T(N)$ rounds. 
\end{proof}

\begin{lemma}
	\label{lem:deterministicLLL}
	There is a deterministic LLL algorithm for range bounded LLL instances with constant dependency degree $d$ and LLL criterion $epd^2<1$ that runs in $O(\poly\log\log n)$ rounds on any graph with at most $n$ nodes if the communication bandwidth is $b=\Theta(n)$.  
\end{lemma}
\begin{proof}
Plug the randomized algorithm of \Cref{lem:LLLboundedRangeRandomized} into the derandomization result of
\Cref{lem:derandomization}. We can apply the lemma because the feasibility  of a computed solution can be checked in $1$ round in \CONGEST. 
\end{proof}



\begin{proof}[Proof of \Cref{thm:rangeBoundedLLL}]
	The goal is to apply the shattering framework  of \cite{FGLLL17}. The pre-shattering phase takes $\poly \Delta+O(\logstar n)$ rounds and afterwards all events with an unset variable induce small components of size $N=\poly(\Delta,\log n)$. Furthermore, each of the small components also forms an LLL instance, but with a slightly worse LLL criterion. Then we apply \Cref{lem:deterministicLLL} to solve all small components in parallel. Next, we describe these steps in detail.
	
	The pre-shattering phase begins with computing a distance-$2$ coloring of the dependency graph $H$ with $O(\Delta^2)$ colors, i.e., a coloring in which each color appears only once in each inclusive neighborhood. In the \LOCAL model this takes $O(\logstar n)$ rounds with Linial's algorithm \cite{linial92}. In the \CONGEST model we can compute such a coloring in  $O(\poly \Delta_H \cdot \logstar n)=O(\logstar n)$ rounds.
	
	Next, we iterate through the $O(\Delta^2)$ color classes to set some of the variables. The unset variables either have the status \emph{frozen} or \emph{non-frozen}. At the beginning all variables are unset and non-frozen. We iterate through the $O(\Delta_H^2)$ color classes and process all vertices (events) of the same color class in parallel. 
	Each non-frozen variable of a processed event $\mathcal{E}$ is sampled; then the node checks how the conditional probabilities of neighboring events have changed. As the LLL is range bounded this step can be implemented in $O(1)$ rounds. If there is an event $\mathcal{E}'$ (possibly $=\mathcal{E}$) whose probability has increased to at least $p'=\sqrt{p}$, its variables are unset and all variables of event $\mathcal{E}'$ are frozen. The next three observations are proven in \cite{FGLLL17} and capture the properties of this pre-shattering phase. 
	\begin{observation}
		\label{obs:smallProb}
		 For each event $\mathcal{E} \in \mathcal{X}$, the probability of $\mathcal{E}$ having at least one unset variable is at most $(d + 1)\sqrt{p}$. Furthermore, this is independent of events that are further than $2$ hops from $\mathcal{E}$.
	\end{observation}
	
	The following result follows with \Cref{obs:smallProb} and the by now standard shattering lemma (see \Cref{lem:shattering}). We do not discuss the details as it has been done in \cite{FGLLL17}.
	\begin{observation}[Small components]
	The connected components in $H$ induced by all events with at least one unset variable are w.h.p. in $n$ of size $N=\poly(\Delta_H)\cdot \log n=O(\log n)$. 
	\end{observation}
The following observation holds as each event with an unset variable fails at most with probability $p'=\sqrt{p}$ when its variables are frozen. It is proven in \cite{FGLLL17}.
	\begin{observation}
		The problem on each connected component induced by events with at least one unset variable is an LLL problem  with criterion $p'(d+1) <1$. 
	\end{observation}

	To complete our proof we apply the deterministic algorithm of \Cref{lem:deterministicLLL} on each component in parallel. Let $U$ be the set of nodes that have its random bits not yet determined. Note that any vertex in $U$ is part of one component. For each $u\in U$ the values of already determined random bits in the $r$-hop ball around $u$ are included in $u$'s input---many nodes already determine their random bits in the pre-shattering phase. Even conditioned on the random bits determined in the pre-shattering phase (formally this provides a promise to the inputs), each instance is a range bounded LLL on a bounded degree graph with at most $N$ nodes and the standard \CONGEST bandwidth is $b= \Omega(\log n)=\Omega(N)$, the runtime is $\poly\log N=\poly\log\log n$. All instances can be dealt with independently since by definition, the connected unsolved components share no events nor variables. 
\end{proof}

The celebrated result by \cite{CP19} says that an LCL problem either cannot be solved faster than in $\Omega(\log n)$ rounds or can be solved with an LLL algorithm.

\begin{lemma}[\cite{CP19}]
	\label{lem:LCLSpeedup}
		Any LCL problem that can be solved with a randomized $o(\log n)$-round \LOCAL algorithm with error probability $<1/n$ on any graph with at most $n$ nodes can be solved via the following procedure: Create a bounded range LLL instance with LLL criterion $p(ed)^{100}<1$, solve the instance, and run a constant time deterministic \LOCAL algorithm that uses the solution of the LLL instance as input. 
\end{lemma}

\begin{proof}
	If we omitted the requirement that the LLL instance has to be range bounded the lemma is directly proven in \cite{CP19}. Thus, we focus on showing that the obtained LLL instance is range bounded. We first explain the setup of the proof of \cite{CP19}. 
	
To solve an LCL with an $o(\log n)$ round algorithm $\mathcal{A}$ in the \LOCAL model, they set up an LLL. The random variables of a node in this LLL consist of the random bits that are needed for an execution of algorithm $\mathcal{A}$ for a carefully chosen constant number of rounds $t_0=O(1)$. 
There is a \emph{bad event} for each node $v\in V$ of the communication network which holds if $\mathcal{A}$ executed with the random bits violates the LCL constraint of $v$.  Thus, the goal of the LLL is to find an assignment of good random bits such that the execution of $\mathcal{A}$ for $t_0$ rounds  does not fail (in solving the LCL problem) at any node.

The core of \cite{CP19} is to show that due to the assumptions on the runtime function one can lie to the algorithm $\mathcal{A}$ about the number of nodes in the input graph, run the algorithm with a 'fake' small value for the number of nodes at the cost of increasing the error probability. They show that there exists an absolute constant (!) value $n_0=O(1)$ that only depends on the LCL and the runtime function $t_{\mathcal{A}}$ of algorithm $\mathcal{A}$ such that the problem of finding good random bits for $\mathcal{A}(n_0)$, that is $\mathcal{A}$ executed for $t_0=t_{\mathcal{A}}(n_0)$ rounds under the belief that the graph has at most $n_0$ vertices, is indeed an LLL problem satisfying the LLL criterion  $p(ed)^{100}<1$. Actually, they show that for any constant $c>1$ there is a choice of $n_0$ such that the criterion $epd^c<1$ holds. 

The only part of \Cref{lem:LCLSpeedup} that is not explicitly proven in \cite{CP19} is the fact that the obtained LLL is range bounded. We next reason that this merely follows as $n_0$ is an absolute constant. We emphasize that this is the part of the proof that is not argued in the previous work but follows as a simple observation from the proof details. The dependency radius (in the communication network) of this LLL is $2(t_{\mathcal{A}}(n_0)+r)$ where $r$ is the checking radius of the LCL. Recall, that we assume that, in a randomized algorithm, the number of random bits used by a node on a graph on at most $n$ nodes can be upper bounded by an arbitrarily fast growing function $h_{\mathcal{A}}(n)$. Thus, the number of random bits (or number of variables) of a node in the LLL is bounded by  $h_{\mathcal{A}}(n_0)=O(1)$.  The LLL is range bounded as any variable, i.e., any random bit, can only take one out of two values. 
		
Once the LLL is solved one can solve the original LCL problem by executing $A(n_0)$ in $t_0=t_{\mathcal{A}}(n_0)=O(1)$ \LOCAL rounds (this execution is deterministic given the random bits from the LLL solution). An $O(1)$-round \LOCAL algorithm can be run in \CONGEST in $\poly \Delta=O(1)$ rounds. 
\end{proof}
The influential result of \Cref{lem:LCLSpeedup} of \cite{CP19} has been used in several other works, e.g., in  \cite{Balliu2021,BGR21}. Note that \cite{Balliu2021} sets up the same LLL with the purpose of designing an efficient \CONGEST algorithm for it; however, \cite{Balliu2021} is restricted to the setting where the input graph is a tree which allows for very different methods of solving the respective LLL and admits, as they show, even an $O(\log\log n)$-round algorithm. 
We combine \Cref{lem:LCLSpeedup} with our LLL algorithm in \Cref{thm:rangeBoundedLLL} to prove \Cref{cor:LCLgap}.

\corLCLgap*

\begin{proof}[Proof of \Cref{cor:LCLgap}]
Let $P$ be an LCL problem with randomized \CONGEST complexity $o(\log n)$, which also implies a \LOCAL algorithm with the same complexity. Now, we use the \CONGEST LLL algorithm from \Cref{thm:rangeBoundedLLL} in the framework of \Cref{lem:LCLSpeedup} to obtain a \CONGEST algorithm with runtime $O(\poly\log \log n)$ for $P$. 
Observe that even though the LLL construction requires running the LLL algorithm from \Cref{thm:rangeBoundedLLL} in a powergraph, since we are dealing with a bounded degree input graph, the dependency degree of the LLL remains a constant.
This approach also requires that we can execute the constant time $\LOCAL$ algorithm at the end in the \CONGEST model, which follows as the maximum degree  is constant for LCLs, the solution to the LLL is constant as the LLL is range bounded and also the output for the LCL problem is constant. 
\end{proof}

The results in this section imply randomized $\poly\log\log n$-round algorithms for classic problems such as $\Delta$-coloring on constant degree graphs (as $\Delta$-coloring is an LCL which has an $o(\log n)$-round \LOCAL algorithm \cite{GHKM18} the result follows along the same lines as \Cref{cor:LCLgap}) and various defective coloring variants for constant degree graphs by modeling them as range bounded LLLs and applying \Cref{thm:rangeBoundedLLL}, see \cite{SuLLL2017} for various such problems and how they can be modeled as LLLs. 

Our network decomposition algorithm with few colors in combination with the \LOCAL model LLL algorithm from \cite{FGLLL17} provides the following theorem. 
Due to the unconstrained message size in the \LOCAL model , it does not require the LLL instances to be range bounded. 
\begin{theorem}
	Let $\lambda\in \mathbb{N}$ be constant. 	There exists a deterministic $n^{2/\lambda}\poly\log n$ round \LOCAL algorithm for LLL instances with criterion $p(ed)^{\lambda}<1$. 
\end{theorem}
\begin{proof}
	Let $H$ be the dependency graph of the LLL. Use \Cref{thm:constantColors} to compute a network decomposition of $H^2$ with $\lambda$ colors and cluster diameter $n^{1/\lambda}\poly\log n$ in $n^{2/\lambda}\poly\log n$ rounds; in the \LOCAL model we can ignore any congestion between the clusters. From \cite[arxiv version, Theorem 3.5]{FGLLL17}, we get that under the LLL criterion $p(ed)^\lambda$ and given a $(\lambda, \gamma)$ network decomposition, we can solve the LLL in $O(\lambda \cdot (\gamma + 1))$ rounds.
By plugging in $\gamma = n^{1/\lambda} \poly \log n$, we observe that the computation of the network decomposition dominates the runtime and we obtain the result.
\end{proof}

\section*{Acknowledgments}
This project was partially supported by the European Union's Horizon 2020 Research and  Innovation Programme under grant agreement no. 755839 (Yannic Maus).


\clearpage
\bibliographystyle{alpha} 
\bibliography{references}

\appendix


\section{Tree aggregation of \cite{GGR20}}

The following result was essentially proven in \cite[Corollary 5.3, arxiv version]{GGR20}.
\begin{corollary}
	\label{cor:treeAggregationBetter}
	Let $G$ be a communication graph on $n$ vertices. Suppose that each vertex of $G$ is
	part of some cluster $\mathcal{C}$ such that each such cluster has a rooted Steiner tree $T_{\mathcal{C}}$ of diameter at most
	$\beta$ and each node of $G$ is contained in at most $\kappa$ such trees. Then, in $O(\max\{1,\kappa/b \}\cdot (\beta + \kappa))$ rounds of the
	CONGEST model with $b$-bit messages, we can perform the following operations for all
	clusters in parallel on all clusters:
	\begin{enumerate}
		\item  Broadcast: The root of $T_{\mathcal{C}}$ sends a $b$-bit message to all nodes in $\mathcal{C}$;
		\item  Convergecast: We have $O(1)$ special nodes $u \in \mathcal{C}$, where each special node starts with a
		separate $b$-bit message. At the end, the root of $T_{\mathcal{C}}$ knows all messages;
		\item  Minimum: Each node $u \in \mathcal{C}$ starts with a non negative $b$-bit number $x_u$ . At the end, the root
		of $T_{\mathcal{C}}$ knows the value of $\min_{u \in \mathcal{C}}x_u$ ;
		\item  Summation: Each node $u \in \mathcal{C}$ starts with a non negative $b$-bit number $x_u$ . At the end, the root
		of $T_{\mathcal{C}}$ knows the value of $\big(\sum_{u\in \mathcal{C}} x_u\big) \mod 2 ^{O(b)}$
	\end{enumerate}
\end{corollary}
\begin{proof}
	 Use $x=\max\{1,\kappa/b \}$ \CONGEST $b$-bit rounds to simulate one \CONGEST $x\cdot b$-bit round. Then, apply \cite[Corollary 5.3, arxiv version]{GGR20} with a slowdown of $x$. 
\end{proof}

\section{Shattering}
We use the following by now standard shattering result for our LLL algorithm in \Cref{sec:LLL}. This lemma and its variants have been applied in a vast amount of results in the area of distributed computing and its application is not limited to solving LLL instances. 
\begin{lemma}[The Shattering Lemma \cite{FGLLL17} cf. \cite{BEPSv3}]
	\label{lem:shattering}
	Let $H = (V,E)$ be a graph with maximum degree $\Delta_H$. Consider
	a process which generates a random subset $B\subseteq V$ where $Pr(v \in B) \leq \Delta^{-c_1}$, for some constant $c_1 \geq 1$,
	and where the random variables ${\bf 1}(v \in B)$ depend only on the randomness of nodes within at most $c_2$
	hops from $v$, for all $v \in V$ , for some constant $c_2 \geq 1$. Moreover, let $Z = H[2c_2+1,4c_2+2]$ be the graph
	which contains an edge between $u$ and $v$ if and only if their distance in $H$ is between $2c_2 + 1$ and $4c_2 + 2$.  Then
	with probability at least $1 - n^{-c_3}$, for any constant $c_3$ satisfying $c_1>c_3+ 4c_2 + 2$, we have the following two
	properties:
	\begin{itemize}
		\item[] {\bf (P1)} $Z[B]$ has no connected component $U$ with $|U| \geq \log_{\Delta}n$.
		\item[] {\bf (P2)} Each connected component of $H[B]$ has size at most $O(\log_{\Delta} n \cdot \Delta^{2c_2})$.
	\end{itemize}
\end{lemma}


\section{Faster Distance-$1$ Color Class Carving with Levels}

\label{app:levels}
In this section we prove the following result. 

\theoremFastBallCarving*

Throughout this section fix the parameter $x>0$ (possibly depending on $n$). In this section we analyze  the trade-off between the number of colors in a resulting network decomposition, the cluster diameter and the runtime in the network decomposition algorithm from \cite{GGR20}. Hence, some of the proofs are along the lines of their counterparts in \cite{GGR20} or streamlined versions of their proofs. Here, we focus on the ball carving for a single color class. 

\subsection{Algorithm: Color Class Carving with Levels}

\subparagraph*{Parameters:} The algorithm and its analysis uses the following parameters. 
\begin{align*}
	\levels & = \log_{4/3}n +1\\
	\phases & = 2\cdot\levels+2\log_{2}n\\
	\totalTokens & = 4\cdot \phases \textbf{ (upper bound on tokens created per node)}\\
	\payPerKill & = \totalTokens\cdot x\\
	\proposalAnswerParameter & = 1/(2\cdot \payPerKill)\\
	\steps & = 1/\proposalAnswerParameter=2\payPerKill=8\cdot x\cdot \phases \\ 
	\kappa & = \phases\cdot \min\{k,\steps\}=\phases\cdot \min\{k,x\cdot \phases\}\\
	\beta & = k\cdot \phases \cdot \steps 
\end{align*}

In the description of the algorithm, we did purposely not describe how the red/blue coloring is performed, how a proposal is implemented and also not how the Steiner Trees are formed as none of these internals is influenced by the choice of $x$. Each step can be implemented similar to the ones in the previous section and have been detailed on in \cite{GGR20}.

\subparagraph*{Algorithm Initialization:} Each node is its own cluster and possesses one token, we run the red/blue coloring with a large distance, each cluster is at level $0$. Clusters can either be \emph{stalled} or \emph{alive}.
We define $t_i(\mathcal{C})$ as the number of tokens of $\mathcal{C}$ at the beginning of the $i$-th
phase and set $t_1(\mathcal{C}) = 1$. 

\subparagraph*{Algorithm:} The algorithm has $\phases$ phases and each phase has $\steps$ steps. 
At the beginning of the phase each cluster wakes up if it was stalling in the previous phase. Then all clusters that have changed their level from the previous phase forget about their previous red/blue color and get a new red/blue color via \cite[arxiv, proposition 4.3]{GGR20}.\footnote{This red/blue coloring provides stronger properties than the red/blue coloring in \Cref{ssec:simpleColorCarving}. These properties are used to prove the separation of clusters. As the runtime of computing the coloring is dominated by the other steps (also for $x=\omega(1)$) we do not detail on the involved specifics.} Clusters that have not changed their levels keep their red/blue color.

\noindent \textbf{One Step:} 

\textit{Proposing:} In each step every (!) vertex $v$ in each cluster $\mathcal{C}$, including vertices in stalling clusters, determines whether there is an adjacent cluster $\mathcal{C'}\neq \mathcal{C}$ with $(lev_{\mathcal{C'}},color_{\mathcal{C'}})<(lev_{\mathcal{C}},color_{\mathcal{C}})$ with priority on the level. If the answer is \emph{yes} node $v$ creates a new token and sends a proposal to  an arbitrary such cluster. 

\textit{Accepting/Rejecting:} For a cluster $\mathcal{C}$ let $P_{\mathcal{C}}$ be the set of vertices that propose to $\mathcal{C}$ in a fixed step. if $|P_{\mathcal{C}}|\geq t({\mathcal{C}})/\proposalAnswerParameter$ holds, the cluster accepts all proposals and it increases its token counter by one for each vertex in $P_{\mathcal{C}}$,  otherwise the cluster kills all vertices in $P_{\mathcal{C}}$ and pays $\payPerKill$ for each killed vertex, that is, it decreases its token counter $t(\mathcal{C})$ by $\payPerKill\cdot P_{\mathcal{C}}$  tokens. If the cluster rejects proposals it becomes \emph{stalling} until the end of the phase.  If a cluster accepts the proposals its token counter increases by $1$ for each accepted vertex. 
At the end of the phase each stalling cluster increments its level by one. 

\textit{Steiner Tree Forming:} At the end of each step accepted proposals join the cluster that they proposed to and Steiner trees are build.
This is the end of the algorithm description. 

\medskip

\textit{Useful Clarifications.} It is essential that also vertices in stalled clusters keep proposing. 
The level of a cluster does not change within a phase and never decreases; it can only be incremented by one in the transition between two consecutive phases. As proven in \cite{GGR20} a cluster that has reached the highest level does not receive proposals anymore as it is non adjacent to any other cluster. In a phase a vertex that initially is in a cluster of level $10$ can go to a cluster of level $9$, then $8$, then $7$ ...and so on.

\subsection{Analysis: Correctness, Steiner Trees and Runtime}


The proof of the following observation uses exactly the same ideas as the corresponding proof in \cite{GGR20} but we measure the potential of vertices instead of the potential of clusters. 
\begin{observation}[small total number of tokens]
	\label{obs:totalNumber}
	Each vertex only changes its cluster $3\cdot \phases+1$ times and  we create at most $\totalTokens:=4\cdot \phases\cdot |S|$ tokens in total. 
\end{observation}
\begin{proof}
	We assign a potential to a node, depending on the phase, the level of its current cluster and the color of its cluster. We denote with $lev_i (\mathcal{C})$ the level of $\mathcal{C}$ during phase $i$. Let $v$ be a node, consider phase $i$ and let $\mathcal{C}$ be the current cluster of node $v$ with color $color_i(\mathcal{C})\in\{0,1\}\}=\{\textsf{red},\textsf{blue}\}$, where $\textsf{red}=0$ and $\textsf{blue}=1$. Note that the phase does not uniquely determine the cluster of node $v$ as nodes change their cluster during a phase. But, for a fixed phase $i$ and cluster $\mathcal{C}$ its color $color_i(\mathcal{C})$ is uniquely determined. 
	Then, let the potential of $v$ be $\Phi(v)=3i-2 lev_i(\mathcal{C})+color_i(\mathcal{C})$
	
	Initially, the potential of a node is $\Phi(v)=O(1)$. The potential of a vertex does not decrease when we go to phase $i+1$,  as the phase counter increases by $1$ (which increases the potential by $3$), $lev_{i+1}(\mathcal{C})\leq lev_i(\mathcal{C})+1$ (which decreases the potential at most by $2$) and $color_{i+1}(\mathcal{C})\leq color_i(\mathcal{C})\pm 1$ (which decreases the potential at most by $1$).  When a vertex changes from  cluster $\mathcal{C}$ to some cluster $\mathcal{C}'$ in a step during in some phase the potential increases at least by $1$, as the phase counter remains the same and either the level remains the same as well and $color_i(\mathcal{C'})=color_i(\mathcal{C})+1$, or $lev_i(\mathcal{C'})\leq lev_i(\mathcal{C})-1$ and $color_i(\mathcal{C'})=color_i(\mathcal{C})\pm 1$.

	The first claim follows as the potential of a vertex is upper bounded by $3\phases +1$. The second part of the claim holds as after the initialization there are $|S|$ tokens in the system and each vertex in $S$ can create at most  $3\phases +1$ tokens, totaling in  $(3\phases +1)|S|+|S|\leq 4\phases\cdot |S|$ tokens, as $\phases\geq 2$. 
\end{proof}

\begin{lemma}[Few vertices die]
	\label{lem:levelsFewDie}
	Total number of vertices that die is $|S|/x$.
\end{lemma}
\begin{proof}
	This requires that you do not pay too much per vertex. 
	Due to \Cref{obs:totalNumber}, each vertex creates at most $\phases$ tokens, that is, there are $4\cdot |S|\cdot \phases$ tokens in total. Whenever a vertex is killed we remove $\payPerKill$ tokens permanently from the algorithm. Thus, the number of vertices that die is upper bounded by 
	\begin{align}
		\frac{\totalTokens}{\payPerKill}\leq \frac{|S|}{x} & \qedhere
	\end{align}
\end{proof}

While, the previous claims did not depend on the number of steps, the choice of \steps{} is crucial for the next lemma. In \Cref{ssec:simpleColorCarving} each (blue) cluster became stalling in each phase. As stalling blue clusters are separated from red clusters in that algorithm, there is synchronous progress towards separation. For the algorithm in this section the separation implied by a stalling cluster is weaker and also the measure of progress is asynchronous. Blue clusters do not have to get stalled in every phase, but we show that a cluster either becomes stalling in a phase or it increases its cluster counter by a factor $\geq 2$.  
A cluster is \emph{finished} once it has reached the last  level $\levels$. 
\begin{lemma}[Progress on Token counter]
	\label{lem:tokenProgress}
	In each phase, each non-finished cluster either doubles its token counter, or it  levels up and keeps at least half of its tokens. Further, the token counter of a cluster is always at least $1$. 
\end{lemma}
\begin{proof}
	Let $t$ be the token counter of a cluster $\mathcal{C}$ at the beginning of some phase. 
	\begin{itemize} 
		\item If $\mathcal{C}$ is not finished and does not level up after the phase its token counter increases by a $(1+\proposalAnswerParameter)$ factor in each step. Using the number of steps in a phase we obtain that its token counter at the end of the phase is at least 
		\begin{align*}
			(1+1/\proposalAnswerParameter)^{\steps}\cdot t\geq (1+\steps/\proposalAnswerParameter)\cdot t\geq  2t.
		\end{align*}
		\item Let $t'\geq t$ be the token counter of cluster $\mathcal{C}$ at the beginning of the step in which it becomes stalling. As the cluster becomes stalling in this step the number of vertices that propose in the step is at most $t'/(\proposalAnswerParameter)$. Thus, the number of tokens removed by \emph{paying} $\payPerKill$ for each rejected proposal is upper bounded by 
		\begin{align*}
			\payPerKill\cdot  t'/(\proposalAnswerParameter) = \payPerKill \cdot t'/(2\payPerKill)\geq t'/2.
		\end{align*}
	Thus, the cluster keeps at least $t'-t'/2=t'/2\geq t/2$ tokens when it becomes stalling in a phase.
		\item The fact that at most $t/2$ tokens are removed from the tokens of a cluster in each step together with the fact that a cluster never kills a vertex if it only has one token (if a cluster has only one token it even accepts if it only gets a single proposal) implies that the token counter is always at least $1$. \qedhere
	\end{itemize}
\end{proof}

With \Cref{lem:tokenProgress} and the total limit on the number of tokens given by \Cref{obs:totalNumber} it is straightforward to show that every cluster has to reach the largest level by the last phase. 

\begin{lemma}[Last level is reached]
	\label{lem:lastLevel}
	At the end of the algorithm each cluster is in the largest level.
\end{lemma}
\begin{proof}
	Due to \Cref{lem:tokenProgress}, whenever a cluster does not level up in a phase, it doubles its number of tokens. When it increases its level (at most by one per phase) it keeps at least half of its tokens. If a cluster only leveled up for $\levels-1$ times, it never finishes and it doubled its value in $\phases-(\levels-1)$ phases, while always having at least one token (\Cref{lem:tokenProgress}). That is, it would posses 
	\begin{align*}
		(1/2)^{(\levels-1)}2^{\phases-(\levels-1)}=2^{\phases-2(\levels-1)}=2^{2\log_2 n}=n^2,
	\end{align*}
	a contradiction to the total number of tokens that exist in the system (due to \Cref{obs:totalNumber} this is upper bounded by $|S|\cdot 4 \phases <n^2$ (this uses that $n$ is larger than a sufficiently large constant, e.g., $300$)). 
\end{proof}

One of the most involved parts in the analysis of the algorithm is to prove that clusters are separated at the end of the execution. None of the steps in \cite{GGR20} is influenced by the choice of $x$ and thus we only state the following lemma. 
The most involved part is to prove that clusters are actually separated at the end of the execution. The proof is along similar lines in \cite{GGR20}. The following lemma is very local in time and only holds for the phase in which a cluster is stalling. It also crucially requires that vertices of stalling clusters keep proposing. 
\begin{lemma}
	\label{lem:levelsSeparation}
	At the end of the algorithm every cluster separated.
\end{lemma}

The fact that Steiner trees are formed properly is proven in \cite{GGR20}. Here, we only restate their observations that influence the runtime of the algorithm and the parameters of the computed decomposition. 

\begin{observation}[Steiner tree growth]
	\label{obs:levelsSteinergrowth}
	In each step, the radius of the Steiner tree of each cluster grows by at most $1$.
\end{observation}

\begin{observation} 
	\label{obs:levelsCongestion}
	Over the whole course of the algorithm each edge is added to at most $\kappa=O(\phases)$ Steiner trees. 
\end{observation}
\begin{proof}
	As we consider $k=1$ an edge is only added to a Steiner tree if one of its endpoints proposes to the other endpoint. Due to \Cref{obs:totalNumber} this can happen at most $O(\phases)$ times for each endpoint. 
\end{proof}

\begin{proof}[Proof of \Cref{thm:ballCarvingFast}]
	Due to \Cref{lem:levelsFewDie} at most $|S|/x$ vertices die in the process and due to \Cref{lem:levelsSeparation} each cluster is separated at the end of the process. 
	Due to \Cref{obs:levelsSteinergrowth} each Steiner tree grows by at most $1$ per step, that is, the weak diameter of the computed decomposition is upper bounded by $\beta=O(\phases\steps)=O(x\cdot \log^2 n)$. Due to \Cref{obs:levelsCongestion} each edge is in at most $O(\phases)=O(\log n)$ Steiner trees. 
	
	The runtime is dominated by the aggregation (to determine the number of proposals for each cluster) in each of the steps. Using the aggregation result from \Cref{cor:treeAggregationBetter} one aggregation step can be performed in $O(\beta+\kappa)=O(x\cdot \log^2 n)$ rounds. There are $\phases\cdot\steps$ of these aggregations, which results in a runtime of 
	\begin{align*}
		O(\phases\cdot\steps\cdot x\cdot \log^2 n))=O(x^2\log^4 n) & \qedhere
	\end{align*}
\end{proof}


\section{Balanced Coloring Algorithms for Cluster Graphs}

\label{app:redblue}
The goal of this section is to devise algorithms that color cluster graphs in a balanced way. These colorings are  used in the network decomposition algorithms in \Cref{app:RG} and \Cref{app:levels}. 

In a cluster collection $\mathcal{C}_1,\ldots,\mathcal{C}_p$ a cluster $\mathcal{C}_i$ is \emph{distance-$k$ isolated} if  for all nodes $v\in \mathcal{C}_i$ and each $\mathcal{C}_j, i\neq j$ and each $w\in \mathcal{C}_j$ the distance of $v$ and $w$ in $G$ is strictly larger than $k$. A component is said to be \emph{distance-$k$ connected} if each pair of nodes in the component are connected by a path in the power graph $G^k$. In this section we prove the following lemma. 
\begin{lemma}
	\label{lem:redBlue}
	Let $k$ be an integer. Consider a graph $G = (V,E)$ and a cluster collection $\mathcal{C}_1,\dots,\mathcal{C}_p$ with congestion $\kappa$, Steiner radius $\beta$ and unique $b$-bit cluster identifiers. Let the bandwidth on the communication graph $G$ be $B\geq b$. There is a deterministic $O(k+(\beta+\kappa)\logstar b + \beta\cdot \kappa)$ round algorithm to color each cluster red or blue such that in each distance-$k$ connected component of clusters with at least $2$ clusters, the fraction of vertices colored blue lies within the interval $[1/2,3/4]$. 
\end{lemma}

The high level idea of the following algorithm stems from \cite[arxiv version Lemma 3.1 and 3.2]{GGR20}. However, to deal with $k>1$ we develop careful whitebox implementations of each step. The most important implementation step is captured by \Cref{lem:pathSelection}.  
	\subparagraph*{Algorithmic steps:} 
	
	\begin{itemize}
		\item \textbf{Step 0 (color isolated clusters):} Each distance-$k$ isolated cluster is colored red. 
		\item \textbf{Step 1 (path selection):} Each cluster $\mathcal{C}$ that is not distance-$k$ isolated selects a (directed) path of length at most $k$ to some other cluster  $\mathcal{C}'\neq \mathcal{C}$. 
		\item \textbf{Step 2  (determine heavy/light clusters):} A cluster is \emph{heavy} if it has at least $11$ incoming paths, otherwise a cluster is \emph{light}. Each cluster determines whether it is heavy or light. 
	\end{itemize}	
	
	\begin{itemize}
		\item \textbf{Step 3 (color children of heavy clusters):}  Heavy clusters color the clusters that selected them as balanced as possible favoring blue (largest possible imbalance $\pm1$). The heavy cluster does not color itself in this step.
		\item \textbf{Step 4 (color heavy clusters):} Color each heavy cluster that has not been colored in step 3, blue. 
		
		\item \textbf{Step 5 (color remaining light clusters):} The graph of remaining light vertices has maximum degree $11$, compute a cover of it with stars of size at least 2, e.g., by computing a maximal independent set of the square of the graph. Then, the leader of the star colors each star as balanced as possible favoring blue. 
	\end{itemize}

	\begin{observation}
		\label{obs:redBlue}
		In any distance-$k$ connected component of clusters the fraction of vertices that is colored blue lies within the interval $[1/2,3/4]$. 
	\end{observation}
	\begin{proof}
		Consider a distance-$k$ connected component with $p$ clusters. The largest relative imbalance in Step 3 occurs when a heavy cluster has $13$ children and $6$ children are colored red and $7$ children are colored blue. This relative imbalance is dominated by the imbalance in step 5 where the largest relative imbalance between red and blue clusters occurs if each star consists of three clusters. In this case a $2/3$ fraction is colored bluer and a $1/3$ fraction is colored red.
		As there are $p$ selected (outgoing) paths in total and each heavy vertex has at least $12$ incoming paths there are at most $|p|/12|$ heavy clusters. Thus, in step 4, at most $p/12$ clusters are colored blue. 
		In total in step 3 at most a $7/13$ fraction of the clusters is colored blue, in step 5 at most a $2/3$ fraction of clusters is colored blue and in step 4 a $p/12$ fraction of all clusters in the component is colored blue. Thus we can upper bound the number of blue clusters by $\max\{7/13,2/3\}\cdot p + p/12=(3/4)\cdot p$. As each step colors at least as many clusters in the component blue as red the claim follows. 
	\end{proof}
	
	\subparagraph*{Implementation and Communication:}
	The implementation of each step crucially relies on the communication structure between the clusters that we build with the next lemma. 
	\begin{lemma}[distance-$k$ cluster connecting structure] 
		\label{lem:pathSelection}
		Let $k\geq 1$. Given a cluster collection with $b$-bit cluster identifiers in a communication network with bandwidth $B\geq b$ there is an $O(k+\beta+ \kappa/B)$ round algorithm that builds BFS trees in $G$ with depth at most $k$ such that
		\begin{enumerate}
			\item the root and all leaves of each BFS tree are nodes of a cluster, 
			\item each non distance-$k$ isolated cluster $\mathcal{C}$ has exactly one vertex $w$ that is a leaf of a BFS tree $T_v$ with root $v$ where $v\in \mathcal{C}'\neq \mathcal{C}$,
			\item each edge is contained in at most four distinct BFS trees.
		\end{enumerate}
	\end{lemma}
	\begin{proof}
	\medskip
		\textbf{Algorithm:} Each vertex $v$ of each cluster $\mathcal{C}$ creates a Steiner tree $T_v$ in which it is the root. To this end it creates a token $(id_{\mathcal{C}},k)$ that holds its cluster ID and a \emph{distance counter}. Then, tokens are forwarded in $k$ iterations with several details in the forwarding process. The equipped IDs in tokens are never altered; the distance counter is decreased by one before a token is forwarded. 
A vertex $u$ only forwards a token over an edge $e=\{u,u'\}$ if in total, over the whole execution of the algorithm,  $u$ has sent fewer than two tokens through edge $e$. Each token sent by $u$ over an edge $e=\{u,u'\}$ has to be equipped with an ID distinct from all tokens that $u$ previously sent over $e$.  Node $u$ does not forward a token over  edge $e$ if $u$ received the token through edge $e$ and if the same token is received twice, only one of them is forwarded. If $u$ receives a token with an ID that is identical to the ID of its own cluster the token is dropped and not forwarded.  Nodes remember their actions to be prepared for potential convergecasts.
		
		Once the $k$ iterations are over, each node $u$ of each cluster $\mathcal{C}'$ that received at least one token equipped with an ID that is distinct from $id_{\mathcal{C}'}$ informs the cluster leader $\ell_{\mathcal{C}'}$ about the fact. 
		Via a $1$-bit message the cluster leader $\ell_{\mathcal{C}'}$ selects one such \emph{leaf} which then selects one of the tokens of the aforementioned type and sends it to its origin $v$ using the edges remembered during the execution. All edges and nodes that are traversed in the convergecast are added to $T_v$. 
		All in-cluster communication is performed with single bit messages and via \Cref{cor:treeAggregationBetter}. We emphasize that all token forwardings happen in the communication network $G$. This ends the description of the algorithm.
		
		\noindent \textbf{Proof of Properties: }
		We begin with Property 1.
		By definition, the root $v$ of BFS tree $T_v$ is part of a cluster. Also a node $u$ can only be a leaf of a BFS tree $T_v$ if it initiated a convergecast, which implies that $u$ is part of a cluster. 
		
		For Property $2$, first observe that a cluster leader chooses at most one node of its cluster to become a BFS leaf. 
		Thus, we only need to show that if a cluster $\mathcal{C}$ has another cluster  in distance at most $k$, at least one vertex of $\mathcal{C}$ receives at least one token.
		We show this by induction over the following invariant: For every non distance-$k$ isolated cluster $\mathcal{C}$ after iteration $i$ there is a node in distance at most $k-i$ from $\mathcal{C}$ that has seen a token from a cluster $\mathcal{C}'\neq \mathcal{C}$.
		
		\textit{Induction start:} As $\mathcal{C}$ is non distance-$k$ isolated there is some cluster $\mathcal{C}'$ in distance at most $k$. As each vertex of $\mathcal{C}'$ creates a token the invariant holds before the first iteration. 

		\textit{Induction step:} Apply the induction hypothesis for iteration $i$ and among the nodes satisfying it let $v_i$ be the node that is closest to $\mathcal{C}$ at the end of iteration $i$, that is, the distance of $v_i$ to $\mathcal{C}$ is $\alpha\leq k-i$ and let $t$ be the token with  $t_{ID}\neq ID_{\mathcal{C}}$ that $v_i$ has seen. If $\alpha<k-i$ we have shown the induction step. So, consider $\alpha=k-i$ and let $v_{i+1}$ be a neighbor of $v_i$ that lies on a shortest path from $v_i$ to $\mathcal{C}$. By the choice of $v_i$, node $v_{i+1}$ has not seen a token from a cluster $\mathcal{C}'\neq \mathcal{C}$. In particular, $v_i$ has forwarded at most one token to $v_{i+1}$ so far (possibly a token with $ID_{\mathcal{C}}$). Thus $t_{ID}$ is not equal to the cluster ID of $v_i$'s cluster, $t$ does not get dropped by $v_i$, and $v_i$ forwards $t$ to $v_{i+1}$ in iteration $i+1$, that is, the hypothesis holds for $i+1$. 

		
		To show Property $3$, notice that by definition, at most $2$ tokens are forwarded over any edge by both endpoints. Hence, at most $4$ different clusters can choose this edge to the BFS with the convergecast.
	\end{proof}

	We next explain how each step can be implemented, given \Cref{lem:pathSelection}. An efficient implementation of several of these steps uses the third property in \Cref{lem:pathSelection}. 
	
	\begin{itemize}
		\item \textbf{Step 0/Step 1:} We first apply \Cref{lem:pathSelection}. Let $\mathcal{C}$ be a cluster and $T_v$ be the BFS tree rooted at $v$ such that $w\in \mathcal{C}$ is a leaf of $T_v$. Then, cluster $\mathcal{C}$ selects the cluster of $v$ and the path between $w$ and $v$ in the BFS tree $T_v$. As there is exactly one such leaf for each non distance-$k$ isolated cluster we obtain the desired path selection. Note that the constant overlap allows us to perform aggregation on the BFS trees efficiently. 

		\item \textbf{Step 2:} Each cluster $\mathcal{C}$ needs to count the number of leaves of BFS trees $T_v$ for which the root $v$ is in $\mathcal{C}$. As the BFS trees have constant overlap, $O(k)$ rounds are sufficient to for each such root to learn the number of leaves in the BFS tree rooted at it. Then, within each cluster we only need to determine whether the total number of leaves in all BFS trees $\{T_v\mid v\in \mathcal{C}\}$ is at least $12$, which can be implemented via an constant bit aggregation within the vertices of the cluster in  $O(k+\beta+\kappa)$ rounds (\Cref{cor:treeAggregationBetter}).
		\item \textbf{Step 3:} See \Cref{claim:heavyGuysColor}.
		\item \textbf{Step 4:} This step requires no communication,  as after step 2 each cluster knows whether it is heavy or light. 
		\item \textbf{Step 5:} See \Cref{claim:starSelection}.
	\end{itemize}

	\begin{claim}
		\label{claim:heavyGuysColor}
		A heavy cluster can color the clusters that selected it with a $\pm 1$ imbalance in $O(k+\beta\cdot \kappa)$ rounds. 
	\end{claim}
	\begin{proof}
		First the leader $\ell_{\mathcal{C}}$ of each heavy cluster $\mathcal{C}$ uses $O(k+\beta\cdot \kappa)$ to count the number of clusters that selected it. Further, each node of the extended Steiner tree of cluster $\mathcal{C}$ receives for each incoming edge of the extended Steiner tree how many clusters (that selected $\mathcal{C}$) can be reached through this edge. This is enough to inform each such cluster about its color such that the total imbalance is $\pm 1$. 
	\end{proof}
	\begin{claim}
		\label{claim:starSelection}
		Step 5 can be implemented in $O((\kappa+\beta)\log^* b)$ rounds in the original communication network.
	\end{claim}
	\begin{proof}
		The dominating step is to color a maximal independent set in the square of the (cluster) graph $H$ induced by light clusters. Due to the definition of light clusters the graph $H$ is of constant degree. In particular, each BFS tree $T_v$ that is used to connect two clusters in graph $H$ has a constant number of leaves. Thus, via \Cref{cor:treeAggregationBetter} (part 2) one round of communication in graph $H$ can be simulated in $(\beta+\kappa)$ rounds in the communication network $G$; in one round only a constant number of connecting vertices of a cluster need to send a message to the root. A maximal independent of $H^2$ set can be computed in $O(\logstar b)$ rounds on $H^2$ which results in $O((\kappa+\beta)\logstar b)$ rounds in the original communication network. 
	\end{proof}		

\begin{proof}[Proof of \Cref{lem:redBlue}]
	Due to the explanation just before this proof all steps can be implemented in $O(k+(\beta+\kappa)\logstar b + \beta\cdot \kappa)$ rounds.  The balanced guarantee of the computed coloring follows with \Cref{obs:redBlue}.	
\end{proof}

\section{Simple Distance-$k$ Color Class Carving}
\label{app:RG}
In this section we prove the following ball carving result. 
\label{ssec:simpleColorCarving}
\theoremSlowBallCarving*

Fix parameters $k\geq 1$ and $x>0$. Define the following parameters for the algorithm:
\begin{align}
	\phases & =\log_{4/3}n +1,  \\
	\proposalAnswerParameter & =x\cdot \phases=O(x\cdot \log n),\\
	\steps & =(\proposalAnswerParameter+1)\cdot \log n=O(x\cdot \log^2 n), \\
	\beta & =k \cdot\phases\cdot\steps=O(k \cdot x\cdot\log^3 n), \\
	\kappa & =2\cdot \phases\cdot \min\{k,\steps\}=O(\log n \cdot \min\{k,x\cdot \log^2 n \}).
\end{align}

\subsection{Simple Algorithm: Distance-$k$ Color Class Carving}

\noindent {\bf Algorithm (for one color class, \Cref{thm:ballCarvingSlow}):}
The algorithm has $\phases$ phases and each phase consists of $\steps$ steps. 

Initialization: Each vertex of $S$ forms its own cluster with a unique $b$-bit cluster identifier. 

\begin{itemize}
	\item \textbf{Phase (for reducing the number of clusters in each distance-$k$ connected component):} 
	\begin{itemize}
		\item \textbf{Coloring clusters:} Use \Cref{lem:redBlue} to color the clusters red and blue in a balanced way.
		\item Repeat $\steps$ times:
		
		{\bf One step:} In the beginning of each step, all nodes in red clusters propose to non-stalling blue clusters within distance at most $k$ and either get accepted or die. We next explain how to determine which red node proposes to which blue cluster. 
		
		\textbf{Distance-$k$ proposal implementation:} We have $k$ iterations where blue nodes in non-stalling clusters send tokens to distance $k$. Each blue node in a non-stalling cluster creates one token which is forwarded  in the \emph{original} graph $G$ (including dead nodes and nodes in $V\setminus S$) for at most $k$ hops. A vertex only forwards the first token that it receives, breaking ties arbitrarily. In particular, blue vertices do not forward any tokens. This builds BFS trees of depth at most $k$ rooted at blue nodes. Each red node that is contained in a BFS tree commits to propose to the blue cluster of the root of the tree. Red nodes that are not in a BFS tree do not propose further.

		\textbf{Accepting/Rejecting Proposals:} Let $P_{\mathcal{C}}$ be the (red) vertices that propose to blue cluster $\mathcal{C}$.   Cluster $\mathcal{C}$ accepts all proposals if its cluster would grow by a $(1+1/\proposalAnswerParameter)$ factor by adding the proposals, that is, if $|\mathcal{C}\cup P_{\mathcal{C}}|\geq (1+1/\proposalAnswerParameter)\cdot |\mathcal{C}|$, otherwise, including the case $|P_{\mathcal{C}}|=0$,  all proposals in $P_{\mathcal{C}}$ are rejected and the vertices in $P_{\mathcal{C}}$ die permanently. They will not be part of any cluster during the algorithm for \Cref{thm:ballCarvingSlow} ever again. However, dead nodes still participate in forwarding information and tokens, in particular, they participate in the token forwarding in the proposal phase. If a blue cluster rejects proposals it becomes \emph{stalling} until the end of the phase.  If a cluster accepts a set of proposals $P_{\mathcal{C}}$ these vertices are added to the cluster and the Steiner trees are updated as described in the next paragraph.
		
		\textbf{Steiner tree building:} For accepted proposals Steiner trees are extending along the BFS trees that are created in the proposal phase. 
		
		\item \textbf{End of Phase:} Stalling blue clusters forget their stalling status, every cluster forgets their color, dead vertices remain dead. 
	\end{itemize}
\end{itemize}

\textit{Clarifications and Observations:} In a phase, there can also be red nodes that for some steps do not propose but then start proposing. In the distance-$k$ proposal a red node cannot get a token from its own cluster as its own cluster is red, i.e., no token originates from it. In one step, red nodes in the same cluster might propose to different blue clusters.  Notice that in the distance-$k$ proposal implementation red nodes are fine with proposing to any blue cluster. Thus, they only need to receive one arbitrary token from an arbitrary blue cluster in distance-$k$, in particular, they do not care which token gets forwarded. The identifier is only used to identify clusters, BFS trees, Steiner trees and in the computation of the balanced coloring.

\subsection{Analysis: Correctness}

\begin{lemma}[few vertices die]
	\label{lem:fewDie}
	The number of vertices that die during the whole execution of the algorithm is upper bounded by $|S|/x$. 
\end{lemma}
\begin{proof}
		Let $S_i$ be the vertices that are alive at the beginning of phase $i=0,\ldots,\phases$. 
		We have $S_0=|S|$. In each phase, each killed vertex can be uniquely charged to one stalling cluster such that each stalling cluster is only charged for $1/\proposalAnswerParameter$ fraction of its vertices. Thus, for $i\in [phases-1]$ we obtain
		\begin{align}
			|S_{i+1}|\geq |S_i|(1-1/\proposalAnswerParameter). 
		\end{align}
		We deduce
		\begin{align}
			S_{\phases}\geq |S|(1-1/\proposalAnswerParameter)^{\phases}\geq |S|\left(1-\frac{\phases}{\proposalAnswerParameter}\right)
		\end{align}  
		Thus there are at most $\phases\cdot  |S|/\proposalAnswerParameter=|S|/x$ dead vertices.
\end{proof}

\begin{lemma}
	\label{lem:blueStalling}
	Every blue cluster is stalling at the end of a phase.
\end{lemma}
\begin{proof}
	In each step, in which a cluster does not become stalling, it grows by a factor of $(1 + 1/\proposalAnswerParameter)$. If it does not become stalling in a phase, then after $\steps=(\proposalAnswerParameter+1)\log n$ iterations we obtain that it has at least 
\begin{align}
	(1+1/\proposalAnswerParameter)^{\steps}> e^{\steps\cdot \frac{1/\proposalAnswerParameter}{1+1/\proposalAnswerParameter}}=e^{\frac{\steps}{\proposalAnswerParameter+1}}= n
\end{align}
vertices, where we used $1+y\geq e^{y/(1+y)}$ for $y>-1$. This is a contradiction. 
\end{proof}

We now show that clusters are separated at the end of the algorithm. 
\begin{lemma}[cluster separation 1]
	\label{lem:stallingSeparation}
	 When a blue cluster becomes stalling in some phase $i$ there is no red vertex in distance at most $k$ from it, and this will be the case until the end of phase $i$.
\end{lemma}
\begin{proof}
	Consider a cluster $\mathcal{C}$ that is stalling in some phase $i$. Consider the step $j$ in which $\mathcal{C}$ became stalling. After step $j$ there is no red vertex in distance-$k$ of $\mathcal{C}$, as in step $j$ every vertex in this neighborhood either proposed to $\mathcal{C}$ and got killed, or proposed to some other cluster such that it either also got killed or joined the cluster, i.e., became blue. In particular, there is no red node in distance-$k$ of $\mathcal{C}$ that does not propose in step $j$. 
	
	As nodes never change their color to red during a phase no node in the distance-$k$ neighborhood of $\mathcal{C}$ is red throughout the phase. 
\end{proof}

 The \emph{cluster graph} $H$ during some point of the execution of the algorithm has a vertex for each cluster $\mathcal{C}$ and an edge between $\mathcal{C}$ and $\mathcal{C}'$ if there are vertices $u\in\mathcal{C}$ and $v\in \mathcal{C}'$ with $dist_G(u,v)\leq k$. 
The connected components of $H$ are called the \emph{distance-$k$ connected components} of the clustering. 

Two clusters $\mathcal{C}$ and $\mathcal{C}'$ are called \emph{distance-$k$ separated} if for all nodes $u\in \mathcal{C}$ and $v\in\mathcal{C}'$ and all paths in $G$ between $u$ and $v$ contain at least $k$ consecutive nodes that are dead or in $V\setminus S$. 
\begin{lemma}\label{lem:componentSeparation}
	At the end of each phase each distance-$k$ connected component of blue clusters is distance-$k$ separated from each distance-$k$ connected component of red clusters. 
\end{lemma}
\begin{proof}
	We want to show that since each blue cluster becomes stalling, it must be the case that a connected component of stalling clusters is distance-$k$ separated from red clusters.
	By Lemma~\ref{lem:blueStalling}, we know that eventually, each blue cluster becomes stalling.
	Consider now a connected component $\mathcal{D}$ of blue clusters and suppose  for a contradiction that there is a red cluster that is in distance $k$ from $\mathcal{D}$.
	Then, there must be a red node that is in distance $k$ from some cluster $\mathcal{C} \in \mathcal{D}$, which is a contradiction by Lemma~\ref{lem:stallingSeparation}.
	Therefore, it must be the case that each distance-$k$ connected component of blue clusters is distance-$k$ separated from red clusters.
\end{proof}
Note that \Cref{lem:blueStalling,lem:stallingSeparation} also imply that each red vertex does not have a blue vertex in distance-$k$ at the end of a phase.

The following lemma does not use any red or blue coloring.
\begin{lemma}
	\label{lem:noComponentIncrease}
	If at some point two clusters $\mathcal{C}$ and $\mathcal{C}'$ are not in the same distance-$k$ connected component of clusters, they will not be in the same component for the rest of  the phase. 
\end{lemma}
\begin{proof}
Every path in $G$ connecting the clusters contains $k$ consecutive vertices that are dead or in $V\setminus S$. So, no alive cluster can ever contain any of these vertices ever again. Thus, they cannot become connected. 
\end{proof}

\begin{lemma}[Connected components decrease]
	\label{lem:connectedDecrease}
	Ignoring all colors, at the end of phase $i$, for each cluster $\mathcal{C}$ the number of clusters in the distance-$k$ connected component of $\mathcal{C}$ is at most $\max\{1,(3/4)^i|S|\}$. 
\end{lemma}
\begin{proof}
	We prove the remaining claim by induction over the phases. For $i=0$, the claim holds as there are initially only $S$ clusters. Fix some phase $i\geq 0$ and assume that the claim holds at the end of phase $i$. 	Due to \Cref{lem:blueStalling} every blue cluster is stalling at the end of phase $i+1$. Thus, due to \Cref{lem:stallingSeparation}, there is no red vertex in the distance-$k$ neighborhood of any blue cluster, which implies that vice versa also for every red cluster there is no blue vertex in the distance-$k$ neighborhood. 
	
	 Now, let $C$ be a connected component of clusters in the cluster graph $H$ at the end of phase $i$ (and also at the beginning of phase $i+1$). By the induction hypothesis we have $|C|\leq \max\{1,(3/4)^i|S|\}$. If $|C|=1$ the claim also holds after phase $i+1$ as due to \Cref{lem:noComponentIncrease} connected components can never increase. If $|C|>1$, then at the beginning of phase $i+1$ the clusters of $C$ are colored red or blue in a balanced way, that is, at most $3/4$ of the clusters in the component have the same color. 
	Let $B$ be the blue clusters of $C$ and $R$ the red clusters of $C$. Due to \Cref{lem:componentSeparation} the clusters in $C$ and $R$ are distance-$k$ separated after the phase. The bound on the size of the connected component at the end of phase $i+1$ follows due to the size of $B$ and $R$ and because \Cref{lem:noComponentIncrease} implies that no cluster can join the component in phase $i+1$. 
\end{proof}

\subsection{Analysis: Steiner Trees}
The correct building of Steiner trees for computing the network decomposition of $G^k$ has been analyzed in \cite{RG20} in the case where nodes in the $i$-th phase use the $i$-th bit of their ID to determine whether they are red or blue. As most observations are not affected at all by the red/blue coloring we only re-visit the observations that are crucial to bound the Steiner tree radius, the congestion and thus also have an influence on the runtime.

\begin{observation}[Steiner tree growth]
	\label{obs:steinerGrowth}
	In each step, the radius of the Steiner tree of each blue cluster grows by at most $k$, while the radius of the Steiner tree of each red cluster does not grow.
\end{observation}

\begin{observation}
	\label{obs:steinerCongestion}
	 In each step an edge is added to at most $2$ Steiner trees and in each phase an edge is added to at most  $\kappa = 2\cdot \min\{k,\steps\}$ Steiner trees. 
\end{observation}
\begin{proof}
	In a step an edge is only added to a Steiner tree if a token traveled through the edge during the distance-$k$ proposal phase. As vertices only forward the first token that they receive there can at most be one token per direction that travel through the edge, that is, in each step an edge joins at most $2$ Steiner trees. 
	
	During a phase, an edge can never be added to a Steiner tree anymore once both of its endpoints become blue. If both endpoints of an edge $e=\{u,v\}$ are red and $e$ is added to a spanning tree both endpoints are blue afterwards and the edge will not be added to another spanning tree in the phase. 
	Thus, we concentrate on edges $e=\{u,v\}$ in which at least one vertex is dead or in $V\setminus S$. Let $u$ be this vertex. The edge $e$ has been added to a spanning tree because a token traveled through. If $e$ is added to the spanning tree again blue has to be closer to $u$ than in the last iteration as otherwise it does not reach any new red vertex.
\end{proof}

Recall, that the congestion of Steiner trees is the maximum number of trees that contain the same edge. 
\begin{observation}
	\label{obs:SteinerTrees}
	At all times the radius of the Steiner trees is at most $\beta$ and each edge of $G$ is in at most  $\kappa=\phases\cdot O(\min\{k,\steps\})$ Steiner trees.
\end{observation}
\begin{proof}
	Due to \Cref{obs:steinerGrowth} the radius of each Steiner tree grows at most by $k$ per step. Thus, the radius throughout all phases is upper bounded by $\beta=O(k\cdot \phases\cdot\steps)=O(k\cdot x\cdot \log^2 n)$.
	
	The congestion bound follows with \Cref{obs:steinerCongestion}.
\end{proof}

\subsection{Analysis: Implementation and Runtime}
The presented algorithm is a $\log n$ factor faster than the variant in \cite{RG20} (that uses unique IDs of size $\poly n$) because we incooperated the faster aggregation methods from \cite{GGR20} for parallel trees when accepting/rejecting proposals. 
\begin{observation}[Runtime per step]
	\label{obs:runtimeStep}
	One step can be implemented in $O\big(k + (\beta +\kappa)\cdot \max\{1,\kappa/b\}\big)$. 
\end{observation}
\begin{proof}
	$O(k)$ to send the proposals. For aggregating the number of proposals we require $Z= (\beta +\kappa)\cdot \max\{1,\kappa/b\}$ via the simultaneous aggregation with overlapping Steiner trees from \Cref{cor:treeAggregationBetter}. Informing nodes whether their proposal is accepted and extending the Steiner trees can also be done in $O(k+Z)$ rounds. 
\end{proof}

\begin{lemma}[Runtime per phase]
	\label{lem:runtimePerPhase}
	The runtime of one phase is $O\big(\beta\cdot \logstar b + \steps\cdot (\beta +\kappa) \cdot  \max\{1,\kappa/b\}\big)$.
\end{lemma}
\begin{proof}
	The red/blue coloring of clusters takes $O(k+(\kappa+\beta)\cdot \logstar b+\kappa\cdot \beta)$ rounds via \Cref{lem:redBlue}.
	Due to \Cref{obs:runtimeStep} each step requires $O(k + \max\{1,\kappa/b\}\cdot (\beta +\kappa))$ rounds, that is, a phase can be implemented in (using $\beta\geq k$ and $\steps\geq \beta$) 
	\begin{align*}
		O\big(k+(\kappa+\beta)\cdot \logstar b+\beta\cdot \kappa\big)+ \steps\cdot O\big(k + (\beta +\kappa)\cdot \max\{1,\kappa/b\}\big) \\
		= O(\beta\cdot \logstar b) + k\cdot \steps + \steps \cdot (\beta +\kappa)\cdot  \max\{1,\kappa/b\} \\
		 = O\big((\kappa+\beta)\cdot \logstar b + \steps\cdot (\beta +\kappa) \cdot  \max\{1,\kappa/b\}\big) & \qedhere
	\end{align*}
\end{proof}

\begin{proof}[Proof of \Cref{thm:ballCarvingSlow}]
	Due to to \Cref{lem:fewDie} at most $|S|/x$ vertices are not clustered. The distance-$k$ cluster separation follows with \Cref{lem:connectedDecrease}. The bound on the Steiner tree radius and congestion is stated in \Cref{obs:SteinerTrees}. The runtime follows by multiplying the runtime per phase from \Cref{lem:runtimePerPhase} by the number of phases. Using $b=\log n$. we get $\max\{1,\kappa/b\}=\max\{1,\log n\cdot \min\{k,x\cdot \log^2  n\} /b\}=\min\{k,x\cdot \log^2  n\}$.
	We obtain
	\begin{align*}
		O(\phases\cdot\beta\cdot \logstar b + \phases\cdot\steps\cdot (\beta+\kappa)\cdot \min\{k,x\cdot \log^2  n\}) \\
		= O(\log^4 n\cdot \logstar b) + O(k\cdot x^2\cdot \log^6 n\cdot \min\{k,x\cdot \log^2  n\}).	 & \qedhere	
	\end{align*} 
\end{proof}

\end{document}